\documentclass[11pt,reqno]{amsart}

\usepackage{mathrsfs}
\usepackage{amsmath}
\usepackage{amssymb}
\usepackage{amsfonts}
\usepackage{amsthm}
\usepackage{a4wide}
\usepackage{graphicx}
\usepackage{color}
\usepackage[sans]{dsfont}
\usepackage{linearA}
\usepackage{pgfkeys}
\usepackage{longtable}
\usepackage{pdflscape}
\usepackage{float}
\usepackage{cancel}
\usepackage{dsfont}
\usepackage{setspace}
\usepackage{array}
\usepackage{hyperref}
\hypersetup{colorlinks=true,allcolors=[rgb]{0,0,0.6}}
\usepackage[english]{babel}
\usepackage[latin1]{inputenc}
\usepackage[T1]{fontenc}
\usepackage{tikz,pgflibraryshapes}
\usepackage{enumitem}
\usepackage{eufrak}
\usepackage[square,numbers,sort&compress]{natbib}
\usepackage[shadow,textwidth=22mm,textsize=tiny]{todonotes}
\usepackage{mathabx}

\newtheorem{theorem}{Theorem}[section]
\newtheorem{lemma}[theorem]{Lemma}
\newtheorem{proposition}[theorem]{Proposition}

\newtheorem{remark}[theorem]{Remark}

\numberwithin{equation}{section}

\author[B. Alvarez]{Benjamin Alvarez}
\address[B. Alvarez]{Institut Elie Cartan de Lorraine \\
Universit{\'e} de Lorraine, 
57045 Metz Cedex 1, France}
\email{benjamin.alvarez@univ-lorraine.fr}

\author[J. Faupin]{J{\'e}r{\'e}my Faupin}
\address[J. Faupin]{Institut Elie Cartan de Lorraine \\
Universit{\'e} de Lorraine, 
57045 Metz Cedex 1, France}
\email{jeremy.faupin@univ-lorraine.fr}

\author[J.-C. Guillot]{Jean-Claude Guillot}
\address[J.-C. Guillot]{CNRS-UMR 7641, Centre de Math\'ematiques Appliqu\'ees, \\
Ecole Polytechnique,
91128 Palaiseau Cedex, France}
\email{guillot@cmapx.polytechnique.fr}

\begin{document}
\bibliographystyle{abbrv} \title[Fermionic Hamiltonians]{Hamiltonian models of interacting fermion fields in Quantum Field Theory}

\begin{abstract}
We consider hamiltonian models representing an arbitrary number of spin $1/2$ fermion quantum fields interacting through arbitrary processes of creation or annihilation of particles. The fields may be massive or massless. The interaction form factors are supposed to satisfy some regularity conditions in both position and momentum space. Without any restriction on the strength of the interaction, we prove that the Hamiltonian identifies to a self-adjoint operator on a tensor product of anti-symmetric Fock spaces and we establish the existence of a ground state. Our results rely on new interpolated $N_\tau$ estimates. They apply to models arising from the Fermi theory of weak interactions, with ultraviolet and spatial cut-offs.
\end{abstract}

\maketitle

\section{Introduction and statement of the main results}

In this paper, we consider a finite number of interacting fermion fields of spin $1/2$. Each field is associated to a different species of particles. The free energy is 
\begin{align}
H_f = \sum_{i=1}^n \int \omega_i( k_i ) b^*_i( \xi_i ) b_i( \xi_i ) d \xi_i , \label{eq:defHf_intro}
\end{align}
where $\xi_i = ( k_i , \lambda_i ) \in \mathbb{R}^3 \times \{ -1/2 , 1/2 \}$, with $k_i$ the momentum variable and $\lambda_i$ the spin variable for the $i^{\mathrm{th}}$ particle. Moreover, $b^\sharp_i( \xi_i )$, where $b^\sharp_i$ stands for $b^*_i$ or $b_i$, are the usual fermionic creation and annihilation operators for the $i^{\mathrm{th}}$ particle. Precise definitions will be given in Section \ref{subsec:model}. The relativistic dispersion relations $\omega_i( k_i )$ are defined by 
\begin{equation*}
\omega_i( k_i ) = \sqrt{ k_i^2 + m_i^2 },
\end{equation*}
with $m_i \ge 0$ the mass of the $i^{\mathrm{th}}$ field. The interaction terms are of the form
\begin{align}
 \int G( \xi_1 , \dots , \xi_n ) b_1^*( \xi_1 ) \dots b_p^*( \xi_p ) b_{p+1}( \xi_{p+1} ) \dots b_n( \xi_n ) d \xi_1 \dots d \xi_n + \mathrm{h.c.}, \label{eq:HI(G)intro}
\end{align}
the first term representing a process where $p$ particles are created while $n-p$ particles are annihilated, with $0\le p \le n$. The second term, which is the hermitian conjugate of the first one, represents the inverse process. In the previous equation, the interaction form factor $G$ is supposed to be square integrable, which makes \eqref{eq:HI(G)intro} a well-defined quadratic form. The interaction Hamiltonian that we consider in this paper, $H_I$, is given by the sum over all possible processes of interaction terms of the form \eqref{eq:HI(G)intro}; see the next section. Formally, the total Hamiltonian is then defined by $H = H_f + H_I$.

Important physical examples of processes that can be described by such a model arise from the Fermi theory of weak interactions. For instance, the weak decay of a muon into an electron, a muon neutrino and an electron antineutrino, or the scattering of an electron with an electron neutrino are well-described at low energy by the Fermi theory (see e.g. \cite[Chapter 2]{GrMu00_01}). Considering the formal Hamiltonian associated to the Lagrangian of the theory, and introducing ultraviolet and spatial cut-offs, one obtains an expression of the form above. Note that, in such processes, four spin $1/2$ fermion fields interact. More details on the Fermi theory of weak interactions, as well as other examples, will be given in Section \ref{subsec:weak}.

Our first main concern will be to show that the formal Hamiltonian $H$ defines a self-adjoint operator on Fock space. An obvious obstruction is that the free energy is only quadratic in the creation and annihilation operators, while the interaction Hamiltonian is of order $n$. Nevertheless, using that all particles involved are fermions, an argument due to Glimm and Jaffe \cite{GlJa77_01} shows that, if all interaction form factors $G$ belong to the Schwartz space of rapidly decreasing functions, then $H_I$ is a bounded operator. In particular, $H$ identifies to a self-adjoint operator. In fact, inspecting the proof of \cite[Proposition 1.2.3]{GlJa77_01}, one can see that, to apply the argument, it is actually sufficient that all $G$ belong to the domain of some power of the harmonic oscillator. Still, this condition imposes a constraint on the infrared behavior of the form factors which is much too strong to cover physically realistic cases. To circumvent this problem, our strategy will consist in applying suitable interpolation arguments in order to obtain refined $N_\tau$ estimates. The latter can then be applied to the abstract model studied here, with only \emph{mild} regularity assumptions on the infrared behavior of the form factors. Note that, as may be expected, the regularity assumptions that we will have to impose will be slightly stronger in the case of massless fields than in the massive case. The refined $N_\tau$ estimates, combined with a perturbation argument, will allow us to obtain the self-adjointness of $H$.

Once the self-adjointness of $H$ is established, our next concern will be to prove the existence of a ground state. In the case where all fields are massive, we will adapt an argument of \cite{DeGe99_01}, while in the more difficult case where some fields are supposed to be massless, we will employ an induction argument and follow the approach of \cite{GrLiLo01_01}. In both cases, the proof will have to rely on the refined $N_\tau$ estimates established previously, instead of the classical ones.

The remainder of this section is organized as follows. In Section \ref{subsec:model}, we define precisely the model considered in this paper, next we state our main results in Section \ref{subsec:results}. Section \ref{subsec:weak} is concerned with examples arising from the Fermi theory of weak interactions.

\subsection{The model}\label{subsec:model}

As mentioned above, we consider in this paper an abstract class of models representing a finite number, $n$, of interacting fermion fields. The total Hilbert space is the tensor product of $n$ antisymmetric Fock spaces,
\begin{align}
\mathcal{H} := \bigotimes_{i=1}^n \mathcal{F} , \qquad \mathcal{F} := \mathbb{C} \oplus \bigoplus_{l=1}^\infty \otimes_a^l \Big ( L^2( \mathbb{R}^3 \times \{-\frac12,\frac12\} ) \Big ). \label{eq:def_Fock}
\end{align}
Here $\otimes_a^l$ stands for the anti-symmetric tensor product. Throughout the paper, we use the notation $\xi = ( k , \lambda ) \in \mathbb{R}^3 \times \{ - 1/2 , 1/2 \}$, i.e. $k$ stands for the momentum variable and $\lambda$ the spin variable. Moreover, to distinguish between the $n$ different species of particles, the variable corresponding to the $i^{\mathrm{th}}$ Fock space will be denoted by $\xi_i = ( k_i , \lambda_i )$.

The free Hamiltonian, acting on $\mathcal{H}$, is the sum of the second quantizations of the free relativistic energy of $n$ particles of masses $m_i \ge 0$, 
\begin{align*}
H_f := \sum_{i=1}^n H_{f,i} ,
\end{align*}
where $H_{f,i}$ acts on the $i^{\mathrm{th}}$ Fock space and is given by
\begin{align*}
H_{f,i} := \mathrm{d} \Gamma ( \omega_i( k_i ) ) , \quad \omega_i( k_i ) := \sqrt{ k_i^2 + m_i^2 } , \quad m_i \ge 0 .
\end{align*}
Let $\mathcal{F}_{\mathrm{fin}}( \mathcal{S} )$ be the subset of $\mathcal{F}$ consisting of vectors $( \varphi_0 , \varphi_1 , \dots ) \in \mathcal{F}$ such that, for all $l \in \mathbb{N}^*$, $\varphi_l \in \otimes_a^l  L^2( \mathbb{R}^3 \times \{-1/2,1/2\})$ identifies to a function in the Schwartz space $\mathcal{S}( \mathbb{R}^{3l} ; \mathbb{C}^{2^l})$ and $\varphi_l=0$ for all but finitely many $l$'s. We recall that, in the sense of quadratic forms on $\mathcal{F}_{\mathrm{fin}}( \mathcal{S} ) \times \mathcal{F}_{\mathrm{fin}}( \mathcal{S} )$, $H_{f,i}$ is given by Equation \eqref{eq:defHf_intro}, 
where $b^*_{i}(\xi_{i})$, respectively $b_{i}(\xi_{i})$, stands for the fermionic creation operator, respectively annihilation operator, acting on the $i^{\text{th}}$ Fock space.  The following anti-commutation relations are supposed to hold:
\begin{align*}
& \{ b_i( \xi_i ) , b^*_i( \xi'_i ) \} = \delta ( \xi_i - \xi'_i ) , \\
& \{ b_i( \xi_i ) , b_i( \xi'_i ) \} = \{ b_i^*( \xi_i ) , b^*_i( \xi'_i ) \} = 0 , \\
& \{ b_i^\sharp( \xi_i ) , b_j^\sharp( \xi_j ) \} = 0 , \quad i < j ,
\end{align*}
for all $i,j \in \{ 1  , \dots , n \}$. In what follows, we use the notation
\begin{align*}
\int f( \xi_i ) d \xi_i = \sum_{ \lambda_i \in \{ - \frac12 , \frac12 \} } \int_{ \mathbb{R}^3 } f( k_i , \lambda_i ) dk_i.
\end{align*}

The interaction Hamiltonian is given by
\begin{align}
H_I := \sum_{p=0}^n \sum_{\{ i_1,\dots,i_n \} \in \mathfrak{I}_p} \big ( H^{(p)}_{I,i_1,\dots,i_n}( G^{(p)}_{i_1,\dots,i_n} ) + \mathrm{h.c.} \big ), \label{eq:defHI}
\end{align}
where we have set
\begin{equation}
\mathfrak{I}_p := \big \{ \{i_1,\dots,i_n\}=\{1,\dots,n\}, \, i_1<\cdots<i_p, \, i_{p+1}<\cdots<i_{n}, \, i_1<i_{p+1} \big \}, \label{eq:def_Ip}
\end{equation} 
 and
\begin{align*}
H_{I,i_1,\dots,i_n}^{(p)} (G^{(p)}_{i_1,\dots,i_n}) := \int G^{(p)}_{i_1,\dots,i_n}( \xi_1 , \dots , \xi_n ) b^*_{i_1}( \xi_{i_1} ) \dots b^*_{i_p}( \xi_{i_p} ) b_{i_{p+1}}( \xi_{i_{p+1}} ) \dots b_{i_n}( \xi_{i_n} ) d \xi_1 \dots d \xi_n .  
\end{align*}
If $p=0$ in \eqref{eq:def_Ip}, it should be understood that the conditions $i_1 < \dots < i_p$ and $i_1 < i_{p+1}$ are empty, and likewise if $p=1$, $p=n-1$ or $p=n$. As mentioned above, physically, the expression of $H_{I,i_1,\dots,i_n}^{(p)} (G^{(p)}_{i_1,\dots,i_n})$ represents an interaction process where a particle of each species labeled by $i_1, \dots , i_p$ is created, while a particle of each species labeled by $i_{p+1}, \dots , i_n$ is annihilated. Summing over $\mathfrak{I}_p$ insures that all possible creation or annihilation processes are considered.  Assuming that $G$ is square integrable, it is not difficult to verify (see \cite[Theorem X.44]{RS}) that $H_{I,i_1,\dots,i_n}^{(p)} (G^{(p)}_{i_1,\dots,i_n})$ defines a quadratic form on
\begin{align*}
\Big ( \hat{\otimes}_{i=1}^n \mathcal{F}_{\mathrm{fin}}( \mathcal{S} ) \Big ) \times \Big ( \hat{\otimes}_{i=1}^n \mathcal{F}_{\mathrm{fin}}( \mathcal{S} ) \Big ),
\end{align*}
where $\hat{\otimes}$ stands for the algebraic tensor product. 

Formally, the total Hamiltonian that we will study is given by
\begin{equation}
H = H_f + H_I. \label{eq:defH_0}
\end{equation}
It is a well-defined quadratic form on $\big ( \hat{\otimes}_{i=1}^n \mathcal{F}_{\mathrm{fin}}( \mathcal{S} ) \big ) \times \big ( \hat{\otimes}_{i=1}^n \mathcal{F}_{\mathrm{fin}}( \mathcal{S} ) \big )$.

\subsection{Results}\label{subsec:results}
Before stating our main results, we introduce some notations. Let
$$
\ldbrack n_1 , n_2 \rdbrack := \{ n_1 , n_1+1 , \dots , n_2 \} , \qquad \ldbrack n_1, n_2 \rdbrack_{ i_0 } := \{ n_1 , n_1+1 , \dots , n_2 \} \setminus \{ i_0 \} ,
$$
for any integers $n_1 < n_2$ and $i_0$. We distinguish massive and massless particles,
\begin{align*}
\ldbrack 1 , n \rdbrack = \ldbrack 1 , n \rdbrack^> \cup \ldbrack 1 , n \rdbrack^0 ,
\end{align*}
with
\begin{align*}
\ldbrack 1 , n \rdbrack^> := \{ i \in \ldbrack 1 , n \rdbrack , m_i > 0 \}, \quad \ldbrack 1 , n \rdbrack^0 := \{ i \in \ldbrack 1 , n \rdbrack , m_i = 0 \}.
\end{align*}
We set in addition 
\begin{align*}
\ldbrack 1 , n \rdbrack^>_{i_0} := \ldbrack 1 , n \rdbrack^> \setminus \{ i_0 \} \quad \text{if } i_0 \in \ldbrack 1 , n \rdbrack^> , \qquad \ldbrack 1 , n \rdbrack^>_{i_0} := \ldbrack 1 , n \rdbrack^> \quad \text{if } i_0 \notin \ldbrack 1 , n \rdbrack^> ,
\end{align*}
and likewise for $\ldbrack 1 , n \rdbrack^0_{i_0}$.

Given operators $A_1, \dots , A_n$, we adopt the convention that
\begin{align*}
\prod_{ i \in \ldbrack 1 , n \rdbrack } A_i u := A_1 \cdots A_n u ,
\end{align*}
for any $u \in \mathcal{D}( A_1 \cdots A_n ) = \{ v \in \mathcal{D}( A_n ) , A_n v \in \mathcal{D} ( A_1 \cdots A_{n-1} ) \}$, where $\mathcal{D}( A )$ stands for the domain of an operator $A$.

Recalling the notation $\xi = ( k , \lambda ) \in \mathbb{R}^3 \times \{ - 1/2 , 1/2 \}$, we let
\begin{align}
\mathrm{h}_i^{j} := - \frac{d^2}{d k_{i,j} ^2} + k_{i,j}^2 \label{eq:def_harm_osc}
\end{align}
be the harmonic oscillator in one-dimension corresponding to the variable $k_{i,j}$, with 
\begin{equation*}
k_i = ( k_{i,1} , k_{i,2} , k_{i,3} ).
\end{equation*}
The corresponding operator acting on the variable $k_{i,j}$ in $\otimes_{i=1}^n L^2( \mathbb{R}^3 \times \{-1/2 , 1/2 \})$ is denoted by the same symbol. In other words, for $G \in \otimes_{i=1}^n L^2( \mathbb{R}^3 \times \{ -1/2 , 1/2 \})$, 
\begin{align*}
( \mathrm{h}_i^{j} G )( \xi_1 , \dots , \xi_n ) = - \frac{\partial^2 G}{\partial k_{i,j}^2} ( \xi_1 , \dots , \xi_n ) + k_{i,j}^2 G ( \xi_1 , \dots , \xi_n ) ,
\end{align*}
with $\xi_i = ( k_{i,1} , k_{i,2} , k_{i,3} , \lambda_i )$.

Our main results are summarized in the following two theorems. The first one shows that $H$ identifies to a self-adjoint operator on $\mathcal{H}$. The second one establishes the existence of a ground state for $H$, under stronger assumptions on the kernels $G^{(p)}_{i_1,\dots,i_n}$.
\begin{theorem}[Self-adjointness]\label{th:main}
Let $i_0 \in \ldbrack 1 , n \rdbrack$ and $\varepsilon > 0$. Suppose that, for all $p \in \ldbrack 0 , n \rdbrack$ and all set of integers $\{i_1,\dots,i_n\} \in \mathfrak{I}_p$,
\begin{align*}
G^{(p)}_{i_1,\dots,i_n} \in \mathcal{D} \Big ( \Big ( \prod_{ \small{ \substack{ i \in \ldbrack 1 , n \rdbrack^{>}_{i_0} \\ j \in \ldbrack 1 , 3 \rdbrack }Ê} } ( \mathrm{h}_i^{j} )^{ \frac12 - \frac{1}{n-1} + \varepsilon } \Big ) \Big ( \prod_{ \small{ \substack{ i \in \ldbrack 1 , n \rdbrack^{0}_{i_0} \\ j \in \ldbrack 1 , 3 \rdbrack }Ê} } ( \mathrm{h}_i^{j} )^{ \frac12 - \frac56 \frac{1}{n-1} +\varepsilon }   \Big ) \Big ).
\end{align*}
Then the quadratic form $H$ defined in \eqref{eq:defH_0} extends to a self-adjoint operator on $\mathcal{H}$ with domain
\begin{align*}
\mathcal{D}( H ) = \mathcal{D}( H_f ).
\end{align*}
Moreover, $H$ is semi-bounded from below and any core for $H_f$ is a core for $H$.
\end{theorem}
For any $\Lambda > 0$, we set $B_\Lambda := \{ \xi = ( k , \lambda ) \in \mathbb{R}^3 \times \{ -1/2 , 1/2 \} , |k| \le \Lambda \}$.
\begin{theorem}[Existence of a ground state]\label{th:GS}
Under the conditions of Theorem \ref{th:main}, suppose in addition that, for all $i' \in \ldbrack 1 , n \rdbrack^0$, $p$ and $i_1,\dots,i_n$ as in the statement of Theorem \ref{th:main}, $1\le r < 2$ and $\Lambda > 0$,
\begin{align}
& \int_{B_\Lambda}  \big| k_{i'} \big|^{-2r} \Big \| \big ( \mathrm{S} G^{(p)}_{i_1,i_2,\dots,i_n} \big )( \cdots , \xi_{i'} , \cdots )  \Big \|_2^r d \xi_{i'} < \infty , \label{eq:cond_GS_1}
\end{align}
and
\begin{align}
& \int_{B_\Lambda}  \big| k_{i'} \big|^{-r} \Big \| \big ( \mathrm{S} ( \nabla_{k_{i'}} G^{(p)}_{i_1,i_2,\dots,i_n} ) \big ) ( \cdots , \xi_{i'} , \cdots )  \Big \|_2^r d \xi_{i'} < \infty , \label{eq:cond_GS_2}
\end{align}
where $\mathrm{S}$ stands for the operator
\begin{align*}
\mathrm{S} := \Big ( \prod_{ \small{ \substack{ i \in \ldbrack 1 , n \rdbrack^{>}_{i_0} \\ j \in \ldbrack 1 , 3 \rdbrack }Ê} } ( \mathrm{h}_i^{j} )^{ \frac12 - \frac{1}{n-1} + \varepsilon } \Big ) \Big ( \prod_{ \small{ \substack{ i \in \ldbrack 1 , n \rdbrack^{0}_{i_0} \setminus \{ i' \} \\ j \in \ldbrack 1 , 3 \rdbrack }Ê} } ( \mathrm{h}_i^{j} )^{ \frac12 - \frac56 \frac{1}{n-1} +\varepsilon }   \Big ) ,
\end{align*}
on $L^2 ( ( \mathbb{R}^3 \times \{ -1/2 , 1/2 \} )^n )$. Then $H$ has a ground state, i.e., $E:=\inf \sigma( H )$ is an eigenvalue of $H$.
\end{theorem}
We emphasize that our results hold without any restriction on the strength of the interaction. This means that, if one introduces a coupling constant $g$ into the model and study the Hamiltonian $H_g = H_f + g H_I$, then $H_g$ is self-adjoint and has a ground state \emph{for all} values of $g \in \mathbb{R}$. In the next section, we give concrete examples insuring that the conditions on the kernels $G^{(p)}_{i_1,i_2,\dots,i_n}$ stated in the previous theorems are satisfied.

\subsection{Applications to mathematical models of the weak interactions}\label{subsec:weak}
As mentioned above, the main examples we have in mind of physical models of the form studied in this paper come from the Fermi theory of weak interactions.

Fermions involved in the Lagrangian of the Standard Model are quarks and leptons. Each fermion is a Dirac particle with spin $1/2$ and is distinct from its antiparticle. There are six quarks, $\mathrm{up} (u)$, $\mathrm{down} (d)$, $\mathrm{strange}(s)$, $\mathrm{charm}(c)$, $\mathrm{bottom}(b)$, $\mathrm{top}(t)$ and six leptons $e_-$, $\nu_e$, $\mu_-$, $\nu_\mu$, $\tau_-$, $\nu_\tau$ where $\nu_e$ (respectively $\nu_\mu$, $\nu_\tau$) is the electron neutrino (respectively the muon neutrino, the tau neutrino).

We are mainly interested in the quark and lepton Lagrangian, which arises by taking the normal Dirac kinetic energy and replacing the ordinary derivative by the covariant one. See \cite{We05_01,We05_02}, \cite{GuKi16_01} and \cite[Section 6.2, (6.11)]{Ka17_01}. This Lagrangian is used to calculate quark and lepton interactions. The full Lagrangian of the fermions is the sum of the lepton electroweak Lagrangian and the quark QCD Lagrangian. It is a finite sum of terms involving two fermions with spin $1/2$ together with gauge bosons. See \cite{Ka17_01}. Hamiltonian models of the weak interaction involving two fermion fields and one boson field have been studied in \cite{AlFapreprint,AsBaFaGu11_01,BaFaGu16_02,BaFaGuappear,BaGu09_01}.

Well known physical examples of interacting fermion fields with spin $1/2$ are given by the Fermi Theory of weak interactions with $V-A$ (Vector-Axial vector) coupling. In these examples, the Fermi Theory is associated to an \emph{effective} low-energy electroweak Lagrangian obtained by the contraction of the propagators of the gauge bosons $W^{\pm}$ and $Z$ in low energy electroweak processes. See e.g. \cite[Chapter 5]{GuKi16_01}. For instance, the weak decay of a muon into an electron, a muon neutrino and an electron antineutrino,
\begin{equation*}
\mu^{-} \to e^{-} + \nu_\mu + \bar{\nu}_e,
\end{equation*}
the scattering of an electron with an electron neutrino,
\begin{equation*}
e^{-} + \nu_e \to \nu_e + e^{-},
\end{equation*}
or the muon production in the scattering of an electron with a muon neutrino,
\begin{equation*}
e^{-} + \nu_\mu \to \mu^{-} + \nu_e,
\end{equation*}
are well-described at low energy by the Fermi Theory of weak interactions. See \cite[Chapter 2]{GrMu00_01}. 

Another fundamental example is the $\beta$ decay, i.e., the weak decay of the neutron. In the Fermi model, the $\beta$ decay is the weak decay of a neutron into a proton, an electron and an electron antineutrino,
\begin{equation*}
n \to p + e^{-} + \bar{\nu}_e.
\end{equation*}
Note that protons and neutrons are baryons, i.e., composite particles. Baryons are composed of three quarks: a neutron is composed of an up quark (u) and two down quarks (d), a proton is composed of two up quarks (u) and one down quark (d). Protons and neutrons may be approximately regarded as bound states of three quarks. In the quark model, the $\beta$ decay is the weak decay of a down quark (d) into a up quark (u), an electron and an electron antineutrino,
\begin{equation*}
d \to u + e^{-} + \bar{\nu}_e.
\end{equation*}
In this model also, four fermions of spin $1/2$ interact.

Recently, the decays of mesons B into mesons D have also been experimentally studied. Mesons are bosons composed of a quark and an antiquark. A meson B is composed of a bottom antiquark $(\overline{b})$ and a quark. A meson D is composed of a charm quark (c) and an antiquark. The following decays have been observed (see \cite{C,CB,DO}):
\begin{equation}\label{1.1}
    B^{-} \rightarrow  D^{0}+ l^{-} + \bar{\nu}_{l} , \qquad    \bar{B}^{0} \rightarrow  D^{+} + l^{-}+  \bar{\nu}_{l},
\end{equation}
and
\begin{equation}\label{1.2}
   B^{+} \rightarrow  \bar{D}^{0} + l^{+}+ \nu_{l}.
\end{equation}
Here $B^{+}= (u\bar{b})$, $B^{-}= (\bar{u}b)$, $ \bar{B}^{0}= (\bar{d}b)$, $D^{0}= (c\bar{u})$, $D^{+}= (c\bar{d})$, $\bar{D}^{0}= (u\bar{c})$, $l^{-}$ (respectively $l^{+}$) is a lepton of negative electric charge (respectively positive electric charge) and $\nu_{l}$, $\bar{\nu}_{l}$ are lepton neutrino and antineutrino. The decays \eqref{1.1} correspond to the transition of a quark b into a quark c,
\begin{equation}\label{1.3}
  b \rightarrow c + l^{-} + \bar{\nu}_{l}.
\end{equation}
The decay \eqref{1.2} is the transition of an antiquark $\bar{b}$ to an antiquark $\bar{c}$,
\begin{equation}\label{1.4}
  \bar{b}\rightarrow \bar{c} + l^{+}+ \nu_{l}.
\end{equation}
The decay \eqref{1.4} is the charge conjugation of the decay \eqref{1.3}. Both involve four fermions of spin $1/2$.

For an example of computation about the six-fermions process $b \rightarrow dq\overline{q}l^{+}l^{-}$, with $q \in (u,d,s)$, see e.g. \cite{HQ}. More generally, all physical processes we have in mind involve an even number of fermions. Nevertheless, for the sake of mathematical generality, we will consider in this paper an arbitrary number $n$ of fermions, with $n$ either even or odd. Some modifications of the proof in the odd case will be required.

In all the previous examples involving four Dirac particles, the formal Hamiltonian obtained from the corresponding Lagrangian is of the form \eqref{eq:defH_0}, with one or two massless fields (assuming that neutrinos are treated as massless, in accordance with the classical form of the Standard Model). After introduction of a (smooth) high-energy cut-off of parameter $\Lambda$ and a spatial cut-off, the kernels obtained from physics can be supposed to be of the form (see \cite{AlFapreprint,BaFaGu16_02,BaFaGuappear})
\begin{align*}
G( \xi_1 , \dots , \xi_4 ) = f_1( k_1 ) \cdots f_4( k_4 ) \delta_{\mathrm{reg}}( k_1 , \dots , k_4 ) ,
\end{align*}
where we disregarded the dependence on the spin variables for simplicity, $\delta_{\mathrm{reg}}( k_1 , \dots , k_4 )$ is a regularization of the Dirac distributions appearing due to momentum conservation, and the $f_i$ are $\mathrm{C}^\infty$ in $\mathbb{R}^3 \setminus \{ 0 \}$, satisfying the estimates
\begin{align}
| \partial^\alpha_{k_{i,j}} f_i( k_i ) | \le C_\alpha |k_i|^{\nu_i-\alpha} \mathds{1}_{ \cdot \le \Lambda }( |k_i| ) , \quad i = 1, \dots , 4 , \quad j = 1 , \dots , 3 , \quad \alpha \in \mathbb{N}. \label{eq:cond_fi} 
\end{align}
Physically, we have that $\nu_i = 0$ for any $i$.

For $n=4$, Theorem \ref{th:main} shows that if all interaction form factors $G$ satisfy
\begin{align}
G \in \mathcal{D} \Big ( \Big ( \prod_{ \small{ \substack{ i \in \ldbrack 1 , n \rdbrack^{>}_{i_0} \\ j \in \ldbrack 1 , 3 \rdbrack }Ê} } ( \mathrm{h}_i^{j} )^{ \frac16 + \varepsilon } \Big ) \Big ( \prod_{ \small{ \substack{ i \in \ldbrack 1 , n \rdbrack^0_{i_0} \\ j \in \ldbrack 1 , 3 \rdbrack }Ê} } ( \mathrm{h}_i^{j} )^{ \frac29 +\varepsilon }   \Big ) \Big ) , \label{eq:cond_n=4}
\end{align}
for some $\varepsilon > 0$, then $H$ is self-adjoint. 
Due to the presence of smooth ultraviolet cut-offs, and assuming that $\delta_{\mathrm{reg}}( k_1 , \dots , k_n )$ is also smooth, the condition \eqref{eq:cond_n=4} is implied by the requirement that, for massive particles, each function $f_i$ belongs to the Sobolev space $H^{1+\varepsilon}( \mathbb{R}^3 )$ for some $\varepsilon>0$ (except possibly the one labelled by $i_0$ that must only belong to $L^2( \mathbb{R}^3 )$), while for massless particles, each function $f_i$ belongs to $H^{4/3+\varepsilon}( \mathbb{R}^3 )$ (again, except possibly the one labelled by $i_0$). Hence we need $\nu_i > -1/2$ in \eqref{eq:cond_fi} for massive particles (i.e. $i \in \ldbrack 1 , n \rdbrack^>_{i_0}$), and $\nu_i > - 1/6$ for massless particles (i.e. $i \in \ldbrack 1 , n \rdbrack^0_{i_0}$), and $\nu_{i_0} > - 3/2$. In particular, the physical case $\nu_i=0$ is covered by our assumptions.

As for the conditions \eqref{eq:cond_GS_1} and \eqref{eq:cond_GS_2} in Theorem \ref{th:GS}, they concern only the massless fields. One can check that they are satisfied provided that $\nu_i \ge 1/2$ in \eqref{eq:cond_fi}, for any $i \in \ldbrack 1 , n \rdbrack^0$. Hence, to prove the existence of a ground state, if massless fields are involved, we need to impose an infrared regularization compared to the physical case. This is due to the method employed \cite{GrLiLo01_01} whose advantage is to allow us to establish the existence of a ground state without any restriction on the strength of the interaction. If one introduces a coupling parameter $g$ into the model and use perturbative methods \cite{BaFrSi98_01,Pi03_01}, it is likely that one can rely on our refined $N_\tau$ estimates to prove the existence of a ground state for $H=H_f + gH_I$ for small enough values of $g$, without imposing any infrared regularization.


\section{Self-adjointness}

In this section we prove that the total Hamiltonian $H$ identifies to a self-adjoint operator, i.e., we prove Theorem \ref{th:main}. The strategy consists in establishing relative bounds of $H_I$ with respect to $H_f$. We begin with relative bounds in the sense of forms, next we turn to operator bounds.

In the following two subsections, we concentrate on a particular term of the interaction Hamiltonian $H_I$ (see \eqref{eq:defHI}) that, for simplicity, we write as
\begin{align}
H_I (G) = \int G( \xi_1 , \dots , \xi_n ) b_1^*( \xi_1 ) \dots b_p^*( \xi_p ) b_{p+1}( \xi_{p+1} ) \dots b_n( \xi_n ) d \xi_1 \dots d \xi_n ,  \label{eq:0}
\end{align}
for some $0 \le p \le n$.

\subsection{Form bounds}\label{subsec:form}

As will be recalled more precisely in the proof of the next lemma, the usual  $N_\tau$ estimates of Glimm and Jaffe \cite{GlJa77_01} show that $H_I (G)$ is relatively form-bounded with respect to. $H_f^{n/2}$. Our first aim is to find suitable conditions on the kernel $G$ such that $H_I$ is relatively form-bounded with respect to lower powers of $H_f$.

We recall that, for any function $f \in L^2( \mathbb{R}^3 \times \{ - 1/2 , 1/2 \} )$,
\begin{equation}
\| b^\sharp( f ) \| = \| f \|_2, \label{eq:bound_f}
\end{equation}
with the usual notations
\begin{align*}
& b^*( f ) = \int f( \xi ) b^*( \xi ) d \xi = \sum_{ \lambda \in \{ -\frac12 , \frac12 \} } \int_{ \mathbb{R}^3 } f( k , \lambda ) b^*_\lambda( k) dk , \\
&  b( f ) = \int \bar{f}( \xi ) b( \xi ) d \xi = \sum_{ \lambda \in \{ -\frac12 , \frac12 \} } \int_{ \mathbb{R}^3 } \bar{f}( k , \lambda ) b_\lambda( k) dk .
\end{align*}
Here $b^\sharp( f )$, $b_\lambda^{\sharp}( f )$ denote the usual fermionic creation and annihilation operators in $\mathcal{F}$. For $g \in L^2( \mathbb{R}^3 )$ and $\lambda \in \{ - 1/2 , 1/2 \}$, the notation $b^\sharp_\lambda( g )$ stands for
\begin{align*}
b^*_\lambda( g ) = \int_{ \mathbb{R}^3 } g( k ) b^*_\lambda( k) dk, \qquad b_\lambda( g ) \int_{ \mathbb{R}^3 } \bar{g}( k ) b_\lambda( k) dk ,
\end{align*}
and hence $\| b^\sharp_\lambda( g ) \| = \| g \|_2$.

We begin with a lemma which is close to Proposition 1.2.3 (b) in \cite{GlJa77_01}. We give a short proof for the convenience of the reader.
\begin{lemma}\label{lm:first}
For all $i_0 \in \ldbrack 1 , n \rdbrack$, $G \in \mathcal{D} ( \prod_{i\in\ldbrack 1 , n \rdbrack_{ i_0 } } \omega_i(k_i)^{-\frac12} )$ and $\varphi \in \mathcal{D}Ê((  \sum_{ i \in \ldbrack 1, n \rdbrack_{i_0} }H_{f,i} )^{ (n-1)/2 } )$, we have that
\begin{align}
\big | \langle \varphi ,ÊH_I(G) \varphi \rangle \big | \le \Big \|Ê \prod_{i\in\ldbrack 1 , n \rdbrack_{ i_0 } } \omega_i(k_i)^{-\frac12} G \Big \|_2 \Big \| \Big ( \sum_{ i \in \ldbrack 1, n \rdbrack_{i_0} }H_{f,i} + 1 \Big )^{ \frac{n-1}{2} } \varphi \Big \|^2. \label{eq:-1}
\end{align}
\end{lemma}
\begin{proof}
Assume that $i_0 \in \ldbrack 1, p \rdbrack$. We write
\begin{align*}
& \big | \langle \varphi ,ÊH_I(G) \varphi \rangle \big | \\
& = \Big | \int \Big \langle \prod_{ i \in \ldbrack 1 , p \rdbrack_{i_0} } b_i( \xi_i ) \varphi , \Big ( \int G( \xi_1 , \dots \xi_n ) b^*_{i_0} ( \xi_{i_0} ) d\xi_{i_0} \Big ) \prod_{ i \in \ldbrack p+1,n \rdbrack } b_i ( \xi_i ) \varphi \Big \rangle \prod_{ i \in \ldbrack 1 , n \rdbrack_{i_0} } d \xi_i \Big |.
\end{align*}
Using the Cauchy-Schwarz inequality and the fact that, for a.e. $\xi_i$, $i \in \ldbrack 1 , n \rdbrack_{i_0}$,
\begin{align*}
\Big \|Ê\int G( \xi_1 , \dots \xi_n ) b^*_{i_0} ( \xi_{i_0} ) d\xi_{i_0} \Big \| = \| G( \xi_1 , \dots , \xi_{i_0-1} , \cdot , \xi_{i_0+1} , \dots , \xi_n ) \|_2,
\end{align*}
(see \eqref{eq:bound_f}), we obtain that
\begin{align}
& \big | \langle \varphi ,ÊH_I(G) \varphi \rangle \big | \notag \\
& = \int \Big \| \prod_{ i \in \ldbrack 1 , p \rdbrack_{i_0} } b_i( \xi_i ) \varphi \Big \|  \| G( \xi_1 , \dots , \xi_{i_0-1} , \cdot , \xi_{i_0+1} , \dots , \xi_n ) \|_2 \Big \| \prod_{ i \in \ldbrack p+1,n \rdbrack } b_i ( \xi_i ) \varphi \Big \| \prod_{ i \in \ldbrack 1 , n \rdbrack_{i_0} } d \xi_i . \label{eq:8}
\end{align}
Now, we observe that
\begin{align*}
\int \prod_{ i \in \ldbrack 1 , p \rdbrack_{i_0} } \omega_i( k_i ) \Big \| \prod_{ i \in \ldbrack 1 , p \rdbrack_{i_0} } b_i( \xi_i ) \varphi \Big \|^2 \prod_{ i \in \ldbrack 1 , p \rdbrack_{i_0} } d \xi_i & = \Big \langle \varphi , \prod_{ i \in \ldbrack 1 , p \rdbrack_{i_0} } H_{f,i} \varphi \Big \rangle \\
& \le \Big \langle \varphi , \Big ( \sum_{ i \in \ldbrack 1 , p \rdbrack_{i_0} } H_{f,i} \Big )^{p-1} \varphi \Big \rangle ,
\end{align*}
and likewise, 
\begin{align*}
\int \prod_{ i \in \ldbrack p+1 , n \rdbrack } \omega_i( k_i ) \Big \| \prod_{ i \in \ldbrack p+1 , n \rdbrack } b_i( \xi_i ) \varphi \Big \|^2 \prod_{ i \in \ldbrack p+1 , n \rdbrack } d \xi_i  & \le \Big \langle \varphi , \Big ( \sum_{ i \in \ldbrack p+1 , n \rdbrack } H_{f,i} \Big )^{n-p} \varphi \Big \rangle .
\end{align*}
Combining this with \eqref{eq:8} and the Cauchy-Schwarz inequality, we deduce that
\begin{align}
& \big | \langle \varphi ,ÊH_I(G) \varphi \rangle \big | \notag \\
& \le \Big \| \prod_{i\in\ldbrack 1 , n \rdbrack_{ i_0 } } \omega_i(k_i)^{-\frac12} G \Big \|_2 \Big \langle \varphi , \Big ( \sum_{ i \in \ldbrack 1 , p \rdbrack_{i_0} } H_{f,i} \Big )^{p-1} \varphi \Big \rangle^{\frac12} \Big \langle \varphi , \Big ( \sum_{ i \in \ldbrack p+1 , n \rdbrack } H_{f,i} \Big )^{n-p} \varphi \Big \rangle^{\frac12} . \label{eq:9}
\end{align}
The estimate \eqref{eq:-1} follows directly from \eqref{eq:9}. The argument is similar in the case where $i_0 \in \ldbrack p+1 , n \rdbrack$.
\end{proof}
\begin{remark}
\label{remark}
The previous proof shows that the following more precise estimate holds:
\begin{align*}
\big | \langle \varphi ,ÊH_I(G) \psi \rangle \big | \le \Big \|Ê \prod_{i\in\ldbrack 1 , n \rdbrack_{ i_0 } } \omega_i(k_i)^{-\frac12} G \Big \|_2 \Big \|  \prod_{ i \in \ldbrack 1 , p \rdbrack_{i_0} } H_{f,i}^{\frac12} \varphi \Big \| \Big \| \prod_{ i \in \ldbrack p+1, n \rdbrack_{i_0} }H_{f,i}^{\frac12}  \psi \Big \| .
\end{align*}
This refined estimate will be useful in the next section.
\end{remark}

Next, we prove another lemma which, in our setting, is a slight improvement of Proposition 1.2.3 (c) in \cite{GlJa77_01} (see also \cite{Am04_01}). The idea is that, if $G$ is sufficiently regular (i.e. belongs to a suitable Schwartz space), then $H_I( G )$ extends to a bounded quadratic form. Our improvement compared to  \cite{GlJa77_01} consists in showing that regularity in all variables \emph{but one} is sufficient to obtain boundedness of $H_I(G)$. This will be important in applications.

Recall that the one-dimensional harmonic oscillator $\mathrm{h}_i^{j}$ has been defined in \eqref{eq:def_harm_osc}. A basis of normalized eigenvectors of $\mathrm{h}_i^{j}$ is denoted by $(e_l)_{ l \in \mathbb{N} }$, so that $\mathrm{h}_i^j e_l = (2l+1) e_l$.
\begin{lemma}\label{lm:second}
For all $s > 1/2$, there exists $\mathrm{C}_{s} > 0$ such that, for all $i_0 \in \ldbrack 1 , n \rdbrack$, $G \in \mathcal{D} ( \prod_{ \small{ \substack{ i \in \ldbrack 1 , n \rdbrack_{i_0} \\ j \in \ldbrack 1 , 3 \rdbrack } } } ( \mathrm{h}_i^{j}  )^s )$ and $\varphi \in \mathcal{H}$,
\begin{align}
\big | \langle \varphi ,ÊH_I(G) \varphi \rangle \big | \le \mathrm{C}_{s} \Big \| \prod_{ \small{ \substack{ i \in \ldbrack 1 , n \rdbrack_{i_0} \\ j \in \ldbrack 1 , 3 \rdbrack } } } ( \mathrm{h}_i^{j}  )^s G \Big \|_2 \| \varphi \|^2 . \label{eq:-2}
\end{align}
\end{lemma}
\begin{proof}
Assume that $i_0 \in \ldbrack 1, p \rdbrack$. We have that
\begin{align*}
& \big | \langle \varphi , ÊH_I(G) \varphi \rangle \big | \\
 &= \Big | \int \Big \langle \prod_{ i \in \ldbrack 1,p \rdbrack } b_i ( \xi_i ) \varphi , G( \xi_1 , \dots , \xi_n ) \prod_{ i \in \ldbrack p+1,n \rdbrack } b_i ( \xi_i ) \varphi \Big \rangle \prod_{ i \in \ldbrack 1 , n \rdbrack } d \xi_i \Big | \\
& = \Big | \sum_{ \small{ \substack { \lambda_i \in \{-\frac12,\frac12\} : \\ i \in \ldbrack 1 , n \rdbrack_{i_0} }  }} \int \Big \langle \varphi , \sum_{ \small{ \substack { l_i^j \in \mathbb{N}  : \\ i \in \ldbrack 1 , n \rdbrack_{i_0} , j \in \ldbrack 1 , 3 \rdbrack } }  } \big \langle e_{l_1^1}Ê\otimes \cdots \otimes e_{l_i^j} \otimes \cdots \otimes e_{l_n^3} , G\big ( (\cdot , \lambda_1) , \cdots , \xi_{i_0} , \cdots , (\cdot , \lambda_n) \big ) \big \rangle_2  \\
&\qquad \qquad \qquad \qquad \quad b^*( \xi_{i_0} ) \prod_{ i \in \ldbrack 1,p \rdbrack_{i_0} } b^*_{i,\lambda_i} ( e_{l_i^j} ) \prod_{ i \in \ldbrack p+1,n \rdbrack } b_{i,\lambda_i} ( e_{l_i^j} ) \varphi \Big \rangle d \xi_{i_0} \Big | ,
\end{align*}
where the subscript in the first sum above means that for each $i \in \ldbrack 1 , n \rdbrack_{i_0}$, we sum over $\lambda_i \in \{-1/2,1/2\}$, and likewise, in the second sum, for each couple $( i , j ) \in \ldbrack 1 , n \rdbrack_{i_0} \times \ldbrack 1 , 3 \rdbrack$, we sum over $l_i^j \in \mathbb{N}$. The scalar product $\langle \cdot , \cdot \rangle_2$ appearing in the right side of the last equality stands for the scalar product in $L^2( \mathbb{R}^{3(n-1)} )$.

Using that $\mathrm{h}_i^{j}$ is self-adjoint and that $\mathrm{h}_i^{j} e_{l_1^1}Ê\otimes \cdots \otimes e_{l_n^3} = ( 2 l_i^j + 1 ) e_{l_1^1}Ê\otimes \cdots \otimes e_{l_n^3}$ for all $i,j$, we obtain that
\begin{align*}
 \big | \langle \varphi ,ÊH_I(G) \varphi \rangle \big | 
& = \Big | \sum_{ \small{ \substack { \lambda_i \in \{-\frac12,\frac12\} : \\ i \in \ldbrack 1 , n \rdbrack_{i_0} }  }} \int \Big \langle \varphi , \sum_{ \small{ \substack { l_i^j \in \mathbb{N}  : \\ i \in \ldbrack 1 , n \rdbrack_{i_0} , j \in \ldbrack 1 , 3 \rdbrack } }  }\prod_{ \small{Ê\substack{ i \in \ldbrack 1 , n \rdbrack_{i_0} \\ j \in \ldbrack 1 , 3 \rdbrack } } } \frac{1}{ ( 2 l_i^j + 1 )^s } \big \langle e_{l_1^1}Ê\otimes \cdots \otimes e_{l_n^3} , \\
& \qquad \qquad \qquad  \prod_{ \small{ \substack{ i \in \ldbrack 1 , n \rdbrack_{i_0} \\ j \in \ldbrack 1 , 3 \rdbrack } } } ( \mathrm{h}_i^{j}  )^s G\big ( (\cdot , \lambda_1) , \cdots , \xi_{i_0} , \cdots , (\cdot , \lambda_n) \big ) \big \rangle_2  \\
& \qquad \qquad \qquad b^*( \xi_{i_0} )  \prod_{ i \in \ldbrack 1,p \rdbrack_{i_0} } b^*_{i,\lambda_i} ( e_{l_i^j} ) \prod_{ i \in \ldbrack p+1,n \rdbrack } b_{i,\lambda_i} ( e_{l_i^j} ) \varphi \Big \rangle d \xi_{i_0} \Big | \\
& \le \sum_{ \small{ \substack { \lambda_i \in \{-\frac12,\frac12\} : \\ i \in \ldbrack 1 , n \rdbrack_{i_0} }  }} \sum_{ \small{ \substack { l_i^j \in \mathbb{N}  : \\ i \in \ldbrack 1 , n \rdbrack_{i_0} , j \in \ldbrack 1 , 3 \rdbrack } }  } \prod_{ \small{Ê\substack{ i \in \ldbrack 1 , n \rdbrack_{i_0} \\ j \in \ldbrack 1 , 3 \rdbrack } } } \frac{1}{ ( 2 l_i^j + 1 )^s } \Big | \Big \langle \varphi , \Big ( \int \big \langle e_{l_1^1}Ê\otimes \cdots \otimes e_{l_n^3} , \\
& \qquad \qquad \qquad \prod_{ \small{ \substack{ i \in \ldbrack 1 , n \rdbrack_{i_0} \\ j \in \ldbrack 1 , 3 \rdbrack } } } ( \mathrm{h}_i^{j} )^s G\big ( (\cdot , \lambda_1) , \cdots , \xi_{i_0} , \cdots , (\cdot , \lambda_n) \big ) \big \rangle_2  b^*( \xi_{i_0} ) d \xi_{i_0} \Big ) \\
& \qquad \qquad \qquad   \prod_{ i \in \ldbrack 1,p \rdbrack_{i_0} } b^*_{i,\lambda_i} ( e_{l_i^j} ) \prod_{ i \in \ldbrack p+1,n \rdbrack } b_{i,\lambda_i} ( e_{l_i^j} ) \varphi \Big \rangle  \Big |.
\end{align*}
Next, by \eqref{eq:bound_f}, we see that the operator into parentheses in the last equation is bounded and satisfies
\begin{align*}
& \Big \|  \int \big \langle e_{l_1^1}Ê\otimes \cdots \otimes e_{l_n^3} , \prod_{ \small{ \substack{ i \in \ldbrack 1 , n \rdbrack_{i_0} \\  j \in \ldbrack 1 , 3 \rdbrack } } } ( \mathrm{h}_i^{j}  )^s G\big ( (\cdot , \lambda_1) , \cdots , \xi_{i_0} , \cdots , (\cdot , \lambda_n) \big ) \big \rangle_{L^2( \mathbb{R}^{3(n-1)} ) } b^*( \xi_{i_0} ) d \xi_{i_0}  \Big \| \\
& =  \Big ( \int  \big | \big \langle e_{l_1^1}Ê\otimes \cdots \otimes e_{l_n^3} , \prod_{ \small{ \substack{ i \in \ldbrack 1 , n \rdbrack_{i_0} \\ j \in \ldbrack 1 , 3 \rdbrack } } } ( \mathrm{h}_i^{j}  )^s G\big ( (\cdot , \lambda_1) , \cdots , \xi_{i_0} , \cdots , (\cdot , \lambda_n) \big ) \big \rangle_2 \big |^2 d \xi_{i_0} \Big )^{\frac12}.
\end{align*}
Combining this with the fact that $\| b^\sharp_{i,\lambda_i} ( e_{l_i^j} ) \| = \| e_{l_i^j} \|_2 = 1$, we deduce that
\begin{align*}
& \big | \langle \varphi ,ÊH_I(G) \varphi \rangle \big | \le \| \varphi \|^2 \sum_{ \small{ \substack { \lambda_i \in \{-\frac12,\frac12\} : \\ i \in \ldbrack 1 , n \rdbrack_{i_0} }  }} \sum_{ \small{ \substack { l_i^j \in \mathbb{N}  : \\ i \in \ldbrack 1 , n \rdbrack_{i_0} , j \in \ldbrack 1 , 3 \rdbrack } }  } \prod_{ \small{Ê\substack{ i \in \ldbrack 1 , n \rdbrack_{i_0} \\ j \in \ldbrack 1 , 3 \rdbrack } } } \frac{1}{ ( 2 l_i^j + 1 )^s } \\
&\qquad \Big ( \int  \big | \big \langle  e_{l_1^1}Ê\otimes \cdots \otimes e_{l_n^3} ,  \prod_{ \small{ \substack{ i \in \ldbrack 1 , n \rdbrack_{i_0} \\ j \in \ldbrack 1 , 3 \rdbrack } } } ( \mathrm{h}_i^{j} )^s G\big ( (\cdot , \lambda_1) , \cdots , \xi_{i_0} , \cdots , (\cdot , \lambda_n) \big ) \big \rangle_2 \big |^2 d \xi_{i_0} \Big )^{\frac12} .
\end{align*}
Applying again the Cauchy-Schwarz inequality, using that $\sum_{l_i^j \in \mathbb{N}} (2 l_i^j+1)^{-2s} < \infty$ since $s > 1/2$, and that the sum over the $\lambda_i$'s is finite, we obtain
\begin{align*}
 \big | \langle \varphi ,ÊH_I(G) \varphi \rangle \big | 
& \le \mathrm{C}_{s} \| \varphi \|^2 \Big ( \sum_{ \small{ \substack { \lambda_i \in \{-\frac12,\frac12\} : \\ i \in \ldbrack 1 , n \rdbrack_{i_0} }  }} \sum_{ \small{ \substack { l_i^j \in \mathbb{N}  : \\ i \in \ldbrack 1 , n \rdbrack_{i_0} , j \in \ldbrack 1 , 3 \rdbrack } }  } \int  \big | \big \langle  e_{l_1^1}Ê\otimes \cdots \otimes e_{l_n^3} , \\
 & \qquad \qquad \qquad   \prod_{ \small{ \substack{ i \in \ldbrack 1 , n \rdbrack_{i_0} \\ j \in \ldbrack 1 , 3 \rdbrack } } } ( \mathrm{h}_i^{j} )^s G\big ( (\cdot , \lambda_1) , \cdots , \xi_{i_0} , \cdots , (\cdot , \lambda_n) \big ) \big \rangle_2 \big |^2 d \xi_{i_0}  \Big )^{\frac12} \\
& = \mathrm{C}_{s} \| \varphi \|^2 \Big \| \prod_{ \small{ \substack{ i \in \ldbrack 1 , n \rdbrack_{i_0} \\ j \in \ldbrack 1 , 3 \rdbrack } } } ( \mathrm{h}_i^{j} )^s G \Big \|_2 ,
\end{align*}
for some positive constant $\mathrm{C}_{s}$. This concludes the proof.
\end{proof}

Now we interpolate the estimates given by Lemmas \ref{lm:first} and \ref{lm:second}. This gives the following proposition.
\begin{proposition}\label{lm:1}
Let $s > 1/2$ and $0 \le \theta \le 1$. There exists a positive constant $\mathrm{C}_{s,\theta}$ such that, for all $i_0 \in \ldbrack 1 , n \rdbrack$, 
\begin{equation*}
G \in \mathcal{D} \Big (  \prod_{ \small{ \substack{ i \in \ldbrack 1 , n \rdbrack^>_{i_0} \\ j \in \ldbrack 1 , 3 \rdbrack }Ê} } ( \mathrm{h}_i^{j} )^{ \theta s}  \prod_{ \small { \substack { i\in \ldbrack 1 , n \rdbrack^0_{i_0} \\ j \in \ldbrack 1 , 3 \rdbrack } } } ( \mathrm{h}_i^{j} )^{\frac{1}{12}+\theta(s-\frac{1}{12})}  \Big ),
\end{equation*}
and $\varphi \in \mathcal{D} ( ( \sum_{ i \in \ldbrack 1, n \rdbrack_{i_0} }H_{f,i} )^{\frac{n-1}{2}( 1 - \theta )} )$, we have that
\begin{align}
& \big |Ê\big \langle \varphi ,ÊH_I ( G ) \varphi \big \rangle \big | \notag \\
&\le \mathrm{C}_{s,\theta} \Big \| \Big ( \prod_{ \small{ \substack{ i \in \ldbrack 1 , n \rdbrack^{>}_{i_0} \\ j \in \ldbrack 1 , 3 \rdbrack }Ê} } ( \mathrm{h}_i^{j} )^{ \theta s} \Big ) \Big ( \prod_{ \small { \substack { i\in \ldbrack 1 , n \rdbrack^0_{i_0} \\ j \in \ldbrack 1 , 3 \rdbrack } } } ( \mathrm{h}_i^{j} )^{\frac{1}{12}+\theta(s-\frac{1}{12})} ) \Big ) G \Big \|_2 \Big \| \Big ( \sum_{ i \in \ldbrack 1, n \rdbrack_{i_0} }H_{f,i} + 1 \Big )^{\frac{n-1}{2}( 1 - \theta )} \varphi \Big \|^2 . \label{eq:estim_lm1}
\end{align}
\end{proposition}

\begin{proof}
To interpolate the estimates given by Lemmas \ref{lm:first} and \ref{lm:second}, we rewrite \eqref{eq:-1} in a weaker version that will be more convenient.
For massive particles, $i \in \ldbrack 1 , n \rdbrack^>$, we have that $\| \omega_i(k_i)^{-\frac12} u \|_2 \lesssim \| u \|_2$, while in the case of massless particles, $i \in \ldbrack 1 , n \rdbrack^0$, we write
\begin{align*}
\| \omega_i(k_i)^{-\frac12} u \|_2 &= \| ( k_{i,1}^2 + k_{i,2}^2 + k_{i,3}^2 )^{-\frac14} u \|_2 \\
&\lesssim \| |k_{i,1}|^{-\frac16} |k_{i,2}|^{-\frac16} |k_{i,3}|^{-\frac16} u \|_2 \\
&\lesssim \| ( | \partial_{k_{i,1}} |^{\frac16}Ê | \partial_{k_{i,2}} |^{\frac16}Ê | \partial_{k_{i,3}} |^{\frac16}Ê ) u \|_2 \\
&\lesssim \Big \| \prod_{j \in \ldbrack 1 , 3 \rdbrack } ( \mathrm{h}_i^{j} )^{\frac1{12}} u \Big \|_2,  \quad i \in \ldbrack 1 , n \rdbrack^0.
\end{align*}
In the first inequality, we used that $abc \lesssim a^3+b^3+c^3$ for any positive numbers $a$, $b$, $c$, and in the second inequality we used Hardy's inequality in $\mathbb{R}$.
Hence \eqref{eq:-1} implies that
\begin{align}
\big |Ê\big \langle \varphi , H_I(G) \varphi \big \rangle \big | \lesssim \Big \|Ê\prod_{ \small { \substack { i\in \ldbrack 1 , n \rdbrack^0_{i_0} \\ j \in \ldbrack 1 , 3 \rdbrack } } } ( \mathrm{h}_i^{j} )^{\frac{1}{12}} G \Big \|_2 \Big \| \Big ( \sum_{ i \in \ldbrack 1, n \rdbrack_{i_0} }H_{f,i} + 1 \Big )^{ \frac{n-1}{2} } \varphi \Big \|^2. \label{eq:3}
\end{align}
 
Now we proceed to interpolation. Let $\tilde G \in L^2$. For $\tilde \varphi  \in \mathcal{H}$, consider the map
\begin{align*}
f:z \mapsto \Big \langle \Big ( & \sum_{ i \in \ldbrack 1, n \rdbrack_{i_0} }H_{f,i} + 1 \Big )^{-\frac{n-1}{2}( 1 - z )} \tilde \varphi , \\
&H_I \Big ( \Big ( \prod_{ \small{ \substack{ i \in \ldbrack 1 , n \rdbrack^{>}_{i_0} \\ j \in \ldbrack 1 , 3 \rdbrack }Ê} } ( \mathrm{h}_i^{j} )^{-zs} \Big ) \Big ( \prod_{ \small { \substack { i\in \ldbrack 1 , n \rdbrack^0_{i_0} \\ j \in \ldbrack 1 , 3 \rdbrack } } } ( \mathrm{h}_i^{j} )^{-\frac{1}{12}-z(s-\frac{1}{12})} ) \Big ) \tilde G \Big ) \Big ( \sum_{ i \in \ldbrack 1, n \rdbrack_{i_0} }H_{f,i} + 1 \Big )^{-\frac{n-1}{2}( 1 - z )} \tilde \varphi \Big \rangle.
\end{align*}
Since the operators $\mathrm{h}_i^{j}$ and $\sum H_{f,i} + 1$ are positive and invertible, one verifies that $f$ is analytic in $\{ z \in \mathbb{C} , 0 < \mathrm{Re}(z) < 1 \}$, and bounded and continuous in $\{ z \in \mathbb{C} , 0 \le \mathrm{Re}(z) \le 1 \}$. Moreover, Equations \eqref{eq:-1} and \eqref{eq:-2} show that 
\begin{align*}
\sup_{ \mathrm{Re}(z) = 0 } | f( z ) | \lesssim \| \tilde G \|_2 \| \tilde \varphi \|^2 , \quad \sup_{ \mathrm{Re}(z) = 1 } | f( z ) | \lesssim \| \tilde G \|_2 \| \tilde \varphi \|^2.
\end{align*}
Applying Hadamard's three lines lemma, we deduce that
\begin{align*}
\sup_{ 0 \le \mathrm{Re}(z) \le 1 } | f( z ) | \lesssim  \| \tilde G \|_2 \| \tilde \varphi \|^2 .
\end{align*}
Taking $\mathrm{Im}(z) = 0$, we obtain that
\begin{align*}
\Big | \Big \langle \Big (  \sum_{ i \in \ldbrack 1, n \rdbrack_{i_0} }H_{f,i} + 1 \Big )^{-\frac{n-1}{2}( 1 - \theta )}  \tilde \varphi \, , \, & H_I \Big ( \Big ( \prod_{ \small{ \substack{ i \in \ldbrack 1 , n \rdbrack^{>}_{i_0} \\ j \in \ldbrack 1 , 3 \rdbrack }Ê} } ( \mathrm{h}_i^{j} )^{- \theta s} \Big )  \Big ( \prod_{ \small { \substack { i\in \ldbrack 1 , n \rdbrack^0_{i_0} \\ j \in \ldbrack 1 , 3 \rdbrack } } } ( \mathrm{h}_i^{j} )^{-\frac{1}{12} - \theta (s-\frac{1}{12})} ) \Big ) \tilde G \Big ) \\
& \Big ( \sum_{ i \in \ldbrack 1, n \rdbrack_{i_0} }H_{f,i} + 1 \Big )^{-\frac{n-1}{2}( 1 - \theta )} \tilde \varphi \Big \rangle \Big | \lesssim \| \tilde G \|_2 \| \tilde \varphi \|^2 ,
\end{align*}
for any $0 \le \theta \le 1$, $\tilde G \in L^2$ and $\tilde \varphi \in \mathcal{H}$. Applying this to 
\begin{align*}
\tilde G = \prod_{ \small{ \substack{ i \in \ldbrack 1 , n \rdbrack^{>}_{i_0} \\ j \in \ldbrack 1 , 3 \rdbrack }Ê} } ( \mathrm{h}_i^{j} )^{ \theta s}  \prod_{ \small { \substack { i\in \ldbrack 1 , n \rdbrack^0_{i_0} \\ j \in \ldbrack 1 , 3 \rdbrack } } } ( \mathrm{h}_i^{j} )^{\frac{1}{12}+\theta(s-\frac{1}{12})} G \qquad \text{and} \qquad \tilde \varphi = ( \sum_{ i \in \ldbrack 1, n \rdbrack_{i_0} }H_{f,i} + 1 )^{\frac{n-1}{2}( 1 - \theta )} \varphi,
\end{align*}
this implies the statement of the lemma.
\end{proof}

\subsection{Operator bounds}

%
In this section we improve the results of Section \ref{subsec:form} by establishing relative bounds of $H_I(G)$ with respect to $H_f$. The first step is to prove the following lemma, using Remark \ref{remark}.

\begin{lemma}\label{lm:first_1}
For all $i_0 \in \ldbrack 1 , n \rdbrack$, $G \in \mathcal{D} ( \prod_{i\in\ldbrack 1 , n \rdbrack_{ i_0 } } \omega_i(k_i)^{-\frac12} )$ and $\varphi \in \mathcal{D}Ê((  \sum_{ i \in \ldbrack 1, n \rdbrack_{i_0} }H_{f,i} )^{ (n-1)/2 } )$, we have that
\begin{align}
\big \|ÊH_I(G) \varphi \big \| \le \Big \|Ê \prod_{i\in\ldbrack 1 , n \rdbrack_{ i_0 } } (1+\omega_i(k_i)^{-\frac12}) G \Big \|_2 \Big \| \Big ( \sum_{ i \in \ldbrack 1, n \rdbrack_{i_0} }H_{f,i} + 1 \Big )^{ \frac{n-1}{2} } \varphi \Big \|. \label{eq:-1_0}
\end{align}
\end{lemma}
\begin{proof}
Suppose that $i_0 \in \ldbrack 1, p \rdbrack$. It follows from Remark \ref{remark} that, for all $\tilde G \in \mathcal{D} ( \prod_{i\in\ldbrack 1 , n \rdbrack_{ i_0 } } \omega_i(k_i)^{-\frac12} )$,
\begin{equation}
\Big ( \prod_{ i \in \ldbrack 1 , p \rdbrack_{i_0} } H_{f,i}^{\frac12} \Big )^{-1} H_I( \tilde G ) \Big ( \prod_{ i \in \ldbrack p+1, n \rdbrack }H_{f,i}^{\frac12} \Big )^{-1} \le \Big \|Ê \prod_{i\in\ldbrack 1 , n \rdbrack_{ i_0 } } \omega_i(k_i)^{-\frac12} \tilde G \Big \|_2. \label{eq:csq_rk}
\end{equation}
Let $G \in \mathcal{D} ( \prod_{i\in\ldbrack 1 , n \rdbrack_{ i_0 } } \omega_i(k_i)^{-\frac12} )$. We claim that
\begin{equation}
H_I( G ) \Big ( \prod_{ i \in \ldbrack 1, n \rdbrack_{i_0} }(H_{f,i}+1)^{\frac12} \Big )^{-1} \in \mathcal{L}( \mathcal{H} ). \label{eq:claim_0}
\end{equation}

To prove \eqref{eq:claim_0}, we write
\begin{align}
& H_I( G ) \Big ( \prod_{ i \in \ldbrack 1, n \rdbrack_{i_0} }(H_{f,i}+1)^{\frac12} \Big )^{-1} \notag \\
& = \Big ( \prod_{ i \in \ldbrack 1 , p \rdbrack_{i_0} } H_{f,i}^{\frac12} \Big )^{-1} \Big ( \prod_{ i \in \ldbrack 1 , p \rdbrack_{i_0} } H_{f,i}^{\frac12} \Big ) H_I( G ) \Big ( \prod_{ i \in \ldbrack 1, n \rdbrack_{i_0} }(H_{f,i}+1)^{\frac12} \Big )^{-1} . \label{eq:ahe1}
\end{align}
Using the pull-through formula $f( H_{f,i} ) b^*(\xi_i) = b^*(\xi_i) f( H_{f,i} + \omega_i(k_i) )$, for any measurable function $f$, we obtain that
\begin{align}
\Big ( \prod_{ i \in \ldbrack 1 , p \rdbrack_{i_0} } H_{f,i}^{\frac12} \Big ) H_I( G ) = \int G( \xi_1 , \dots , \xi_n ) & b_1^*( \xi_1 ) \dots b_p^*( \xi_p ) b_{p+1}( \xi_{p+1} ) \dots b_n( \xi_n ) \notag \\
& \Big ( \prod_{ i \in \ldbrack 1 , p \rdbrack_{i_0} } (H_{f,i} + \omega_i(k_i) )^{\frac12} \Big ) d \xi_1 \dots d \xi_n . \label{eq:ahe2}
\end{align}
Moreover, we have that
\begin{equation}
\big \|Ê( \omega_i( k_i ) + 1 )^{-\frac12} (H_{f,i} + \omega_i(k_i) )^{\frac12} (H_{f,i}+1)^{-\frac12} \big \| \le 1. \label{eq:ahe3}
\end{equation}
Inserting \eqref{eq:ahe2} and \eqref{eq:ahe3} into \eqref{eq:ahe1}, it is not difficult to deduce that
\begin{align*}
& \Big \|ÊH_I( G ) \Big ( \prod_{ i \in \ldbrack 1, n \rdbrack_{i_0} }(H_{f,i}+1)^{\frac12} \Big )^{-1} \Big \| \\
& \le \Big \| \Big ( \prod_{ i \in \ldbrack 1 , p \rdbrack_{i_0} } H_{f,i}^{\frac12} \Big )^{-1} H_I \Big ( \prod_{ i \in \ldbrack 1, p \rdbrack_{i_0} } ( \omega_i( k_i ) + 1 )^{\frac12} |G| \Big ) \Big ( \prod_{ i \in \ldbrack p+1, n \rdbrack }(H_{f,i}+1)^{\frac12} \Big )^{-1} \Big \|.
\end{align*}

Applying \eqref{eq:csq_rk} to $\tilde G = \prod_{ i \in \ldbrack 1, p \rdbrack_{i_0} } ( \omega_i( k_i ) + 1 )^{\frac12} |G|$ proves \eqref{eq:claim_0} and establishes \eqref{eq:-1_0} in the case where $i_0 \in \ldbrack 1 , p \rdbrack$. The case where $i_0 \in \ldbrack p+1,n \rdbrack$ can be treated in the same way.
\end{proof}
Now, as in the previous section, we use an interpolation argument to obtain the following relative bound.
\begin{proposition}\label{prop:relative_bound}
Let $s > 1/2$ and $0 \le \theta \le 1$. There exists a positive constant $\mathrm{C}_{s,\theta}$ such that, for all $i_0 \in \ldbrack 1 , n \rdbrack$, 
\begin{equation*}
G \in \mathcal{D} \Big (  \prod_{ \small{ \substack{ i \in \ldbrack 1 , n \rdbrack^{>}_{i_0} \\ j \in \ldbrack 1 , 3 \rdbrack }Ê} } ( \mathrm{h}_i^{j} )^{ \theta s}  \prod_{ \small { \substack { i\in \ldbrack 1 , n \rdbrack^0_{i_0} \\ j \in \ldbrack 1 , 3 \rdbrack } } } ( \mathrm{h}_i^{j} )^{\frac{1}{12}+\theta(s-\frac{1}{12})} ) \Big ),
\end{equation*} 
and $\varphi \in \mathcal{D} ( ( \sum_{ i \in \ldbrack 1, n \rdbrack_{i_0} }H_{f,i} )^{\frac{n-1}{2}( 1 - \theta )} )$, we have that
\begin{align}
& \big \|ÊH_I ( G ) \varphi \big \| \notag \\
&\le \mathrm{C}_{s,\theta} \Big \| \Big ( \prod_{ \small{ \substack{ i \in \ldbrack 1 , n \rdbrack^{>}_{i_0} \\ j \in \ldbrack 1 , 3 \rdbrack }Ê} } ( \mathrm{h}_i^{j} )^{ \theta s} \Big ) \Big ( \prod_{ \small { \substack { i\in \ldbrack 1 , n \rdbrack^0_{i_0} \\ j \in \ldbrack 1 , 3 \rdbrack } } } ( \mathrm{h}_i^{j} )^{\frac{1}{12}+\theta(s-\frac{1}{12})} ) \Big ) G \Big \|_2 \Big \| \Big ( \sum_{ i \in \ldbrack 1, n \rdbrack_{i_0} }H_{f,i} + 1 \Big )^{\frac{n-1}{2}( 1 - \theta )} \varphi \Big \| . \label{eq:estim_lm1_2}
\end{align}
\end{proposition}
\begin{proof}
It follows from Lemma \ref{lm:second} that, for all $s > 1/2$, there exists $\mathrm{C}_{s} > 0$ such that, for all $i_0 \in \ldbrack 1 , n \rdbrack$, $G \in \mathcal{D} ( \prod_{ \small{ \substack{ i \in \ldbrack 1 , n \rdbrack_{i_0} \\ j \in \ldbrack 1 , 3 \rdbrack } } } ( \mathrm{h}_i^{j}  )^s )$ and $\varphi  , \psi\in \mathcal{H}$,
\begin{align}
\big |Ê\big \langle \varphi , H_I(G) \psi \big \rangle \big | \le \mathrm{C}_{s} \Big \| \prod_{ \small{ \substack{ i \in \ldbrack 1 , n \rdbrack_{i_0} \\ j \in \ldbrack 1 , 3 \rdbrack } } } ( \mathrm{h}_i^{j}  )^s G \Big \|_2 \| \varphi \| \|Ê\psi \|.  \label{eq:new_estim}
\end{align}
Considering the map
\begin{align*}
f:z \mapsto \Big \langle \varphi \Big| H_I \Big ( \Big ( \prod_{ \small{ \substack{ i \in \ldbrack 1 , n \rdbrack^{>}_{i_0} \\ j \in \ldbrack 1 , 3 \rdbrack }} } ( \mathrm{h}_i^{j} )^{-zs} \Big ) \Big ( \prod_{ \small { \substack { i\in \ldbrack 1 , n \rdbrack^0_{i_0} \\ j \in \ldbrack 1 , 3 \rdbrack } } } ( \mathrm{h}_i^{j} )^{-\frac{1}{12}-z(s-\frac{1}{12})} ) \Big ) \tilde G \Big ) \Big ( \sum_{ i \in \ldbrack 1, n \rdbrack_{i_0} }H_{f,i} + 1 \Big )^{-\frac{n-1}{2}( 1 - z )}  \psi \Big \rangle ,
\end{align*}
on $\{ z \in \mathbb{C}, 0 \leq \mathrm{Re}(z) \leq 1 \}$, it suffices to proceed in the same way as in Lemma \ref{lm:1}, using Hadamard's three lines lemma together with \eqref{eq:-1_0} and \eqref{eq:new_estim}.
\end{proof}
\begin{remark}\label{rk:refined_rel_bound}
The constants $\mathrm{C}_{s,\theta}$ in Propositions \ref{lm:1} and \ref{prop:relative_bound} depend on the positive masses $m_i$, $i \in \ldbrack 1 , n \rdbrack^>$. More precisely, inspecting the proof, one can write
\begin{equation*}
\mathrm{C}_{s,\theta} = \tilde{\mathrm{C}}_{s,\theta} ( \prod_{i\in\ldbrack 1 , n \rdbrack^>} m_i^{-\theta/2} ) ,
\end{equation*}
where $\tilde{\mathrm{C}}_{s,\theta}$ is independent of $m_i$. In the next section, to establish the existence of a ground state for the Hamiltonian $H$, it will be important to have relative bounds that hold uniformly in the masses $m_i$, for $i$ in some subset $\mathrm{I} \subset \ldbrack 1 , n \rdbrack^>$. To obtain such bounds, it suffices to modify the proof of Proposition \ref{lm:1} by replacing the estimate $\| \omega_i(k_i)^{-1/2} u \|_2 \le m_i^{-1/2} \| u \|_2$, for $i \in \mathrm{I}$, by $\| \omega_i(k_i)^{-1/2} u \|_2 \le \| |k_i|^{-1/2} u \|_2$. This leads to the following more precise relative bounds
\begin{align}
& \big \|ÊH_I ( G ) \varphi \big \| \le \tilde{\mathrm{C}}_{s,\theta} \min_{\mathrm{I} \subset \ldbrack 1 , n \rdbrack^{>}_{i_0} } \Big \| \Big ( \prod_{ \small{ \substack{ i \in \ldbrack 1 , n \rdbrack^>_{i_0} \setminus \mathrm{I} \\ j \in \ldbrack 1 , 3 \rdbrack }Ê} } m_i^{-\frac{\theta}{2}}( \mathrm{h}_i^{j} )^{ \theta s} \Big ) \Big ( \prod_{ \small { \substack { i\in \ldbrack 1 , n \rdbrack^0_{i_0} \cup \mathrm{I} \\ j \in \ldbrack 1 , 3 \rdbrack } } } ( \mathrm{h}_i^{j} )^{\frac{1}{12}+\theta(s-\frac{1}{12})} ) \Big ) G \Big \|_2 \notag \\
&\qquad \qquad \qquad \qquad \times \Big \| \Big ( \sum_{ i \in \ldbrack 1, n \rdbrack_{i_0} }H_{f,i} + 1 \Big )^{\frac{n-1}{2}( 1 - \theta )} \varphi \Big \| , \label{eq:estim_lm1_2_4}
\end{align}
where $\tilde{\mathrm{C}}_{s,\theta}$ is independent of the masses $m_i$'s.
\end{remark}

\subsection{Self-adjointness of $H$}

Using Proposition \ref{prop:relative_bound}, we are now able to prove Theorem \ref{th:main}.
%
%
%
%
\begin{proof}[Proof of Theorem \ref{th:main}]
Let $\varepsilon > 0$ and consider a term $H^{(p)}_{I,i_1,\dots,i_n}( G^{(p)}_{i_1,\dots,i_n} )$ occurring in the sum defining $H_I$, see \eqref{eq:defHI}. Possibly changing variables, we can assume without loss of generality that $H^{(p)}_{I,i_1,\dots,i_n}( G^{(p)}_{i_1,\dots,i_n} )$ is given by an expression of the form \eqref{eq:0}, hence, to shorten notations, we write $H^{(p)}_{I,i_1,\dots,i_n}( G^{(p)}_{i_1,\dots,i_n} ) = H_I(G)$, with 
\begin{align*}
G \in \mathcal{D} \Big ( \Big ( \prod_{ \small{ \substack{ i \in \ldbrack 1 , n \rdbrack^{>}_{i_0} \\ j \in \ldbrack 1 , 3 \rdbrack }Ê} } ( \mathrm{h}_i^{j} )^{ \frac12 - \frac{1}{n-1} + \varepsilon } \Big ) \Big ( \prod_{ \small{ \substack{ i \in \ldbrack 1 , n \rdbrack^{0}_{i_0} \\ j \in \ldbrack 1 , 3 \rdbrack }Ê} } ( \mathrm{h}_i^{j} )^{ \frac12 - \frac56 \frac{1}{n-1} +\varepsilon }   \Big ) \Big ).
\end{align*}

Applying Proposition \ref{prop:relative_bound} with $\theta = 1 - \frac{2(1-\varepsilon)}{n-1}$ and $s = 1/2 + \kappa$, with $\kappa$ small enough, we obtain that
\begin{align}
\big \|ÊH_I ( G ) \varphi \big \| \lesssim & \Big \|  \Big ( \prod_{ \small{ \substack{ i \in \ldbrack 1 , n \rdbrack^{>}_{i_0} \\ j \in \ldbrack 1 , 3 \rdbrack }Ê} } ( \mathrm{h}_i^{j} )^{ \frac12 - \frac{1}{n-1} + \varepsilon } \Big ) \Big ( \prod_{ \small{ \substack{ i \in \ldbrack 1 , n \rdbrack^{0}_{i_0} \\ j \in \ldbrack 1 , 3 \rdbrack }Ê} } ( \mathrm{h}_i^{j} )^{\frac12 - \frac56 \frac{1}{n-1} +\varepsilon }  \Big ) G \Big \|_2 \big \| (H_f+1)^{1-\varepsilon} \varphi \big \| , \label{eq:6}
\end{align}
for all $\varphi \in \mathcal{D}  ( H_f^{1-\varepsilon}  )$. Next, we observe that
\begin{align}
\big \| (H_f+1)^{1-\varepsilon} \varphi \big \| &= \big \langle \varphi , (H_f+1)^{2-2\varepsilon} \varphi \big \rangle^{\frac12} \notag \\
&\le \big ( \mu^2 \big \langle \varphi , (H_f+1)^{2} \varphi \big \rangle + \mathrm{C}_\mu \|Ê\varphi \|^2 \big )^{\frac12} \notag \\
&\le \mu \big \| H_f \varphi \big \| + \mathrm{C}_\mu \|Ê\varphi \| , \label{eq:rha1}
\end{align}
for any $\mu>0$, the first inequality being a consequence of Young's inequality.

Combining \eqref{eq:6} and \eqref{eq:rha1}, since $\mu>0$ can be fixed arbitrarily small, we deduce that $H_I( G )$ is relatively $H_f$-bounded with relative bound $0$. Since the other terms in the sum occurring in \eqref{eq:defHI} can be treated in the same way, we deduce from the Kato-Rellich theorem that, indeed, $H$ extends to a self-adjoint operator satisfying $\mathcal{D}( H ) = \mathcal{D}( H_f )$. Semi-boundedness of $H$ and the fact that any core for $H_f$ is a core for $H$ are other consequences of the Kato-Rellich theorem.
\end{proof}

\section{Existence of a ground state}

In this section, we prove the existence of a ground state for the Hamiltonian $H$ defined by Theorem \ref{th:main}, i.e., we prove Theorem \ref{th:GS}. In Section \ref{subsec:massive}, we begin by studying the simplest case where all fermion fields are supposed to be massive. Next, in Section \ref{subsec:one_massless}, we consider the case where excatly one field is massless, and, in Subsection \ref{subsec:massless}, we consider the general case, using an induction in $l \in \ldbrack 1 , n \rdbrack$, where $l$ represents the number of massless fields involved.

\subsection{Models with only massive fields}\label{subsec:massive}

In this section, we suppose that the masses of all the particles are positive. We set
\begin{equation*}
m := \min_{i \in \ldbrack 1 , n \rdbrack}  m_i >0.
\end{equation*}
We prove the existence of a ground state for $H$ by adapting a method due to \cite{DeGe99_01} (see also \cite{Ta14_01}). The proofs follows closely Section 4.2 of \cite{BaFaGu16_02}, the main difference being that we have to use the relative bounds of Proposition \ref{prop:relative_bound} instead of the usual $N_\tau$ estimates.

Recall that the Fock space $\mathcal{F}$ has been defined in \eqref{eq:def_Fock}. Let
\begin{equation*}
U : \mathcal{F} \to \mathcal{F} \otimes \mathcal{F},
\end{equation*}
be defined by
\begin{align*}
 U \Omega = \Omega \otimes \Omega , \quad  U b^*( f \oplus g ) = ( b^*( f ) \otimes \mathds{1} + ( -1 )^{ N } \otimes b^*( g ) ) U,
\end{align*}
where $N=\mathrm{d}\Gamma( \mathds{1} )$ denotes the number operator in $\mathcal{F}$. The operator $U$ extends by linearity to a unitary map on $\mathcal{F}$. Let $j_0 \in
\mathrm{C}^\infty( [ 0 , \infty ) ; [ 0 , 1 ] )$ be such that $j_0 \equiv 1$ on $[ 0 , 1/2 ]$ and $j_0 \equiv 0$ on $[ 1 , \infty )$, and let $j_\infty = \sqrt{1 - j_0^2Ê}$. For $R > 0$, let $j_\sharp^R := j_\sharp ( y^2 / R^2 )$ on $L^2( \mathbb{R}^3 \times \{ -1/2 , 1/2 \} )$, where $j_\sharp$ stands for $j_0$ or $j_\infty$ and $y = \mathrm{i} \nabla_k$. Let
\begin{align*}
j^R : L^2( \mathbb{R}^3 \times \{ -\frac12 , \frac12 \} ) &\to L^2( \mathbb{R}^3 \times \{ -\frac12 , \frac12 \} ) \oplus L^2( \mathbb{R}^3 \times \{ -\frac12 , \frac12 \} ) , \quad j^R( f ) := ( j_0^R ( f ) , j_\infty^R ( f ) ).
\end{align*}
We then define
\begin{equation*}
\check{\Gamma}(j^R) : \mathcal{F} \to \mathcal{F} \otimes \mathcal{F} , \quad \check{\Gamma}(j^R) := U \Gamma( j^R ) ,
\end{equation*}
where, as usual, for an operator $a$ on $L^2( \mathbb{R}^3 \times \{-1/2,1/2\} )$, the operator $\Gamma( a )$ on $\mathcal{F}$ is defined by its restriction to $\otimes_a^l L^2( \mathbb{R}^3 \times \{-1/2,1/2\} )$ as $\Gamma( a ) = a \otimes \dots \otimes a$, and $\Gamma(a)\Omega=1$.

The ``extended'' Hilbert space is 
\begin{equation*}
\mathcal{H}^{\mathrm{ext}} := \mathcal{H} \otimes \mathcal{H},
\end{equation*}
where, recall, $\mathcal{H} = \otimes_{i=1}^n \mathcal{F}$. The operator $\check{\Gamma}_R: \mathcal{H} \to \mathcal{H}^{\mathrm{ext}}$ is defined by
\begin{equation*}
\check{\Gamma}_R := \otimes_{i=1}^n \check{\Gamma}( j^R ),
\end{equation*}
and the extended Hamiltonian, acting on $\mathcal{H}^{\mathrm{ext}}$, is
\begin{equation*}
H^{\mathrm{ext}} := H \otimes \mathds{1} + \mathds{1} \otimes H_f.
\end{equation*}
Under the conditions of Theorem \ref{th:main}, one verifies, by adapting the proof of that theorem in a straightforward way, that $H_I \otimes \mathds{1}$ is relatively $H^{\mathrm{ext}}$ bounded with relative bound $0$.

Recall that, for $p \in \ldbrack 0 , n \rdbrack$, the set $\mathfrak{J}_p$ has been defined in \eqref{eq:def_Ip}.
\begin{lemma}
\label{LemCommutation}
Let $i_0 \in \ldbrack 1 , n \rdbrack$ and $\varepsilon > 0$. Suppose that, for all $p \in \ldbrack 0 , n \rdbrack$ and all set of integers $\{ i_1,\dots,i_n \} \in \mathfrak{I}_p$,
\begin{align*}
G^{(p)}_{i_1,\dots,i_n} \in \mathcal{D} \Big (  \prod_{ \small{ \substack{ i \in \ldbrack 1 , n \rdbrack_{i_0} \\ j \in \ldbrack 1 , 3 \rdbrack }Ê} } ( \mathrm{h}_i^{j} )^{ \frac12 - \frac{1}{n-1} + \varepsilon }   \Big ).
\end{align*}
Then, for any $\chi \in \mathrm{C}_0^{\infty}(\mathbb{R})$ we have that
\begin{equation}
\chi(H^{\mathrm{ext}})\check{\Gamma}_R-\check{\Gamma}_R \chi(H) \to 0 ,  \quad R \to \infty.
\end{equation}
\end{lemma}
\begin{proof}
The proof can be adapted from that of \cite[Lemma 4.3]{BaFaGu16_02}. The main difference is that we have to use the relative bound of Proposition \ref{prop:relative_bound} instead of the usual $N_\tau$ estimates. Namely, considering, as in the proof of Theorem \ref{th:main}, a particular term of $H_I$ of the form $H^{(p)}_{I,i_1,\dots,i_n}( G^{(p)}_{i_1,\dots,i_n} ) = H_I(G)$, with $H_I(G)$ given by \eqref{eq:0}, one can compute
\begin{align*}
&( H_I(G)\otimes \mathds{1} ) \check{\Gamma}_R - \check{\Gamma}_R H_I(G) \\
 & =\int \Big(1-\prod_{i=1}^{n}j_0(\frac{x_i^2}{R^2})\Big)G(\xi_1,\dots,\xi_n)b^{*,0}_1(\xi_1) \dots b^{*,0}_p(\xi_p)b^0_{p+1}(\xi_{p+1}) \dots b^0_n(\xi_n) d \xi_1 \cdots d\xi_n\\
& + \underset{\exists j, \alpha_j \neq 0 }{\sum_{ \{\alpha_i\} \in \{0,\infty\}^n }} \int \Big(\prod_{i=1}^{n}j_{\alpha_i}(\frac{x_i^2}{R^2})\Big)G(\xi_1,\dots,\xi_n)b^{*,\alpha_1}_1(\xi_1) \dots b^{*,\alpha_p}_p(\xi_p)b^{\alpha_{p+1}}_{p+1}(\xi_{p+1}) \dots b^{\alpha_n}_n(\xi_n) \\
&\qquad \qquad \qquad \quad d \xi_1 \cdots d\xi_n.
\end{align*}
Here we have set $x_i = \mathrm{i} \nabla_{k_i}$, $b^{\sharp,0}_i := b^{\sharp}_i \otimes \mathds{1}$ and $b^{\sharp,\infty}_i = (-1)^{N_i} \otimes b^{\sharp}$, where $b^\sharp$ stands for $b$ or $b^*$ and $N_i=\int b^*(\xi_i) b( \xi_i ) d \xi_i$ is the number operator in the $i^{\text{th}}$ Fock space. The subscript $\{\alpha_i\} \in \{0,\infty\}^n$, $\exists j, \alpha_j \neq 0$ means that, for each term of the sum, at least one of the creation of annihilation operator $b_{i}^{\sharp,\alpha_i}$ is equal to $b_{i}^{\infty,\alpha_i}$.

Proceeding as in the proof of Proposition \ref{prop:relative_bound}, one then verifies that
\begin{align}
& \big \|Ê\big ( ( H_I(G)\otimes \mathds{1} ) \check{\Gamma}_R - \check{\Gamma}_R H_I(G)\big ) \varphi \big \| \notag \\
&\le \mathrm{C}_{s,\theta} \Big \{ \Big \| \Big ( \prod_{ \small{ \substack{ i \in \ldbrack 1 , n \rdbrack_{i_0} \\ j \in \ldbrack 1 , 3 \rdbrack }Ê} } ( \mathrm{h}_i^{j} )^{ \theta s} \Big ) \Big(1-\prod_{i=1}^{n}j_0(\frac{x_i^2}{R^2})\Big) G \Big \|_2 \notag \\
& \qquad \qquad + \underset{\exists j, \alpha_j \neq 0 }{\sum_{\{\alpha_i\} \in \{0,\infty\}^n }} \Big \| \Big ( \prod_{ \small{ \substack{ i \in \ldbrack 1 , n \rdbrack_{i_0} \\ j \in \ldbrack 1 , 3 \rdbrack }Ê} } ( \mathrm{h}_i^{j} )^{ \theta s} \Big ) \Big(\prod_{i=1}^{n}j_{\alpha_i}(\frac{x_i^2}{R^2})\Big) G \Big \|_2 \Big \} \notag \\
& \qquad \times  \Big \| \Big ( \sum_{ i \in \ldbrack 1, n \rdbrack_{i_0} }( H_{f,i} \otimes \mathds{1} + \mathds{1} \otimes H_{f,i} ) + 1 \Big )^{(n-1)( 1 - \theta )} \varphi \Big \| , \label{eq:estim_lm1_2_3}
\end{align}
for any $s>1/2$ and $0\le\theta\le 1$. Fixing $\theta$ and $s$ as in the proof of Theorem \ref{th:main}, and using that $H_{f,i} \otimes \mathds{1} + \mathds{1} \otimes H_{f,i}$ is relatively $H^{\mathrm{ext}}$ bounded, one deduces from the previous estimate that
\begin{align}
& \big \|Ê\big ( ( H_I(G)\otimes \mathds{1} ) \check{\Gamma}_R - \check{\Gamma}_R H_I(G)\big ) ( H^{\mathrm{ext}} + i )^{-1} \big \| \notag \\
&\le \mathrm{C}_{s,\theta} \Big \{ \Big \| \Big ( \prod_{ \small{ \substack{ i \in \ldbrack 1 , n \rdbrack_{i_0} \\ j \in \ldbrack 1 , 3 \rdbrack }Ê} } ( \mathrm{h}_i^{j} )^{ \frac12 - \frac{1}{n-1} + \varepsilon } \Big ) \Big(1-\prod_{i=1}^{n}j_0(\frac{x_i^2}{R^2})\Big) G \Big \|_2 \notag \\
& \qquad \qquad + \underset{\exists j, \alpha_j \neq 0 }{\sum_{\{\alpha_i\} \in \{0,\infty\}^n }} \Big \| \Big ( \prod_{ \small{ \substack{ i \in \ldbrack 1 , n \rdbrack_{i_0} \\ j \in \ldbrack 1 , 3 \rdbrack }Ê} } ( \mathrm{h}_i^{j} )^{ \frac12 - \frac{1}{n-1} + \varepsilon } \Big ) \Big(\prod_{i=1}^{n}j_{\alpha_i}(\frac{x_i^2}{R^2})\Big) G \Big \|_2 \Big \} .
\end{align}
Using pseudo-differential calculus together with the fact that $G$ belongs to the domain of $\prod_{ \small{ \substack{ i \in \ldbrack 1 , n \rdbrack_{i_0} \\ j \in \ldbrack 1 , 3 \rdbrack }Ê} } ( \mathrm{h}_i^{j} )^{ \frac12 - \frac{1}{n-1} + \varepsilon }$, it is not difficult to see that the right-hand-side of the previous equation goes to $0$, as $R \to \infty$. Since the other terms occurring in the definition of $H_I$ can be treated in the same way, the rest of the proof is a straightforward adaptation of \cite[Lemma 4.3]{BaFaGu16_02}.
\end{proof}
Given Lemma \ref{LemCommutation}, one deduces the location of the essential spectrum of $H$ as stated in the following proposition. Again, the proof is a straightforward adaptation of a corresponding result in \cite{BaFaGu16_02} (see \cite[Theorem 3.5]{BaFaGu16_02}). Details are left to the reader.
\begin{proposition}\label{prop:massive}
Let $i_0 \in \ldbrack 1 , n \rdbrack$ and $\varepsilon > 0$. Suppose that, for all $p \in \ldbrack 0 , n \rdbrack$ and all set of integers $\{i_1,\dots,i_n\} \in \mathfrak{I}_p$,
\begin{align*}
G^{(p)}_{i_1,\dots,i_n} \in \mathcal{D} \Big ( \prod_{ \small{ \substack{ i \in \ldbrack 1 , n \rdbrack_{i_0} \\ j \in \ldbrack 1 , 3 \rdbrack }Ê} } ( \mathrm{h}_i^{j} )^{ \frac12 - \frac{1}{n-1} + \varepsilon }  \Big ).
\end{align*}
The essential spectrum of $H$ is given by
\begin{equation*}
\sigma_{\mathrm{ess}}(H)=[E+m ,\infty),
\end{equation*}
where $E=\inf \sigma( H )$ and $m = \min_{i \in \ldbrack 1 , n \rdbrack}(m_i)>0$. In particular, $E$ is a (discrete) eigenvalue of $H$.
\end{proposition}

\subsection{Models with one massless field}\label{subsec:one_massless}

In this section, we suppose that one field is massless, while all the other ones are massive. To fix ideas, we suppose that $m_1=0$ and $m = \min_{ i \in \ldbrack 2 , n \rdbrack } m_i > 0$. It will be convenient to write the total Hamiltonian $H$ in \eqref{eq:defH_0} as $H=H_{m_1=0}$, in other words,
\begin{equation*}
H_{m_1=0} = H_{f,m_1=0} + H_I = \mathrm{d}\Gamma( |k_1| )+\sum_{i=2}^n \mathrm{d}\Gamma( \sqrt{k_i^2+m_i^2} ) + H_I , \quad m_i > 0 , \quad 2 \le i \le n ,
\end{equation*}
where $H_I$ is given by \eqref{eq:defHI}. Obviously the situation is identical if $m_i=0$ for some $i \in \ldbrack 1 , n \rdbrack$ while all the other $m_i$'s are positive.

In order to prove that the Hamiltonian $H_{m_1=0}$ has a ground state, we follow the strategy of \cite{GrLiLo01_01}. First, we approximate $H_{m_1=0}$ by a family of operators $H_{m_1}$, where
\begin{equation*}
H_{m_1} = H_{f,m_1} + H_I = \sum_{i=1}^n \mathrm{d}\Gamma( \sqrt{k_i^2+m_i^2} ) + H_I , \quad m_i > 0 , \quad 1 \le i \le n ,
\end{equation*}
then we let $m_1 \to 0$. Proposition \ref{prop:massive} shows that $H_{m_1}$ has a ground state $\Phi_{m_1}$. The strategy then consists in showing that $\Phi_{m_1}$ converges strongly, as $m_1\to 0$, to a (non-vanishing) ground state of $H_{m_1=0}$.

We set $E_{m_1}:=\inf \sigma( H_{m_1} )$ and $E_{m_1=0}:=\inf \sigma( H_{m_1=0} )$.
\begin{proposition}\label{minimsequence}
Let $i_0 \in \ldbrack 1 , n \rdbrack$ and $\varepsilon > 0$. Suppose that, for all $p \in \ldbrack 0 , n \rdbrack$ and all set of integers $\{i_1,\dots,i_n\} \in \mathfrak{I}_p$,
\begin{align*}
G^{(p)}_{i_1,\dots,i_n} \in \mathcal{D} \Big ( \Big ( \prod_{ \small{ \substack{ i \in \ldbrack 2 , n \rdbrack_{i_0} \\ j \in \ldbrack 1 , 3 \rdbrack }Ê} } ( \mathrm{h}_i^{j} )^{ \frac12 - \frac{1}{n-1} + \varepsilon } \Big ) \Big ( \prod_{ \small{ \substack{ i \in \ldbrack 1 , 1 \rdbrack_{i_0} \\ j \in \ldbrack 1 , 3 \rdbrack } }Ê}  ( \mathrm{h}_i^{j} )^{ \frac12 - \frac56 \frac{1}{n-1} +\varepsilon }   \Big ) \Big ) .
\end{align*}
Then $E_{m_1} \to E_{m_1=0}$ as $m_1 \downarrow 0$. Moreover,
\begin{equation}
E_{m_1=0} = \lim_{ m_1 \downarrow 0 } \langle \Phi_{m_1} , H_{m_1=0} \Phi_{m_1} \rangle. \label{eq:Em1=lim}
\end{equation}
\end{proposition}
\begin{proof}
For $m_1>m'_1>0$, we have that $H_{m_1} \ge H_{m'_1} \ge H$. This implies that the map $m_1 \mapsto E_{m_1}$ is non-decreasing and bounded below by $E_{m_1=0}$. Hence there exists $E^*$ such that 
\begin{equation*}
\lim_{m_1 \downarrow 0} E_{m_1} =: E^* \geq E_{m_1=0}.
\end{equation*}

We prove that $E^* \leq E_{m_1=0}$. Let $\delta>0$ and let $\varphi \in \mathcal{D}( H_{m_1=0} )$ be a normalized vector such that $\langle \varphi , H_{m_1=0} \varphi \rangle \le E_{m_1=0} + \delta$. Note that
\begin{equation*}
H_{m_1} \leq \mathrm{d}\Gamma (|k_1| + m_1) + \sum_{i \in \ldbrack 2 , n \rdbrack }\mathrm{d}\Gamma ( \sqrt{ k_i^2 + m_i^2 } ) + H_I \leq H_{m_1=0} + m_1 N_1 ,
\end{equation*}
where $N_1 = \int b^*_1(\xi_1) b_1( \xi_1 ) d \xi_1$ is the number operator in the first Fock space.

Let $\tilde {\varphi}_\ell =\mathds{1}_{[0,\ell]}(N_1) \varphi$. Clearly, since $N_1$ commutes with $\mathrm{d}\Gamma( (k_i^2+m_i^2)^{1/2} )$, $i\in \ldbrack 1 , n \rdbrack$, we have that
\begin{equation*}
\lim_{\ell \to \infty} \big \langle \tilde {\varphi}_\ell , H_{f,m_1=0} \tilde {\varphi}_\ell \big \rangle = \big \langle \varphi , H_{f,m_1=0} \varphi \big \rangle.
\end{equation*}
Moreover, applying Proposition \ref{lm:1}, we see that there exist $a>0$ and $b>0$ such that
\begin{equation*}
\big |Ê\big \langle \tilde {\varphi}_\ell - \varphi , H_I (\tilde {\varphi}_\ell-\varphi)\big \rangle \big | \leq a \big \langle \tilde {\varphi}_\ell-\varphi , H_{f,m_1=0} (\tilde {\varphi}_\ell-\varphi) \big \rangle + b \| \tilde {\varphi}_\ell - \varphi \|^2,
\end{equation*}
which tends to zero as $\ell\to\infty$. This shows that, fixing $\ell$ large enough, we have
\begin{equation*}
\langle \tilde{\varphi}_\ell - \varphi , H_{m_1=0} (\tilde{\varphi}_\ell - \varphi) \rangle \leq \delta ,
\end{equation*}
and therefore
\begin{align*}
E^* & \leq \langle \tilde{\varphi}_\ell , H_{m_1} \tilde{\varphi}_\ell \rangle \\
& \leq  \langle \tilde{\varphi}_\ell , H_{m_1=0} \tilde{\varphi}_\ell \rangle + m_1  \langle \tilde{\varphi}_\ell , N_1 \tilde{\varphi}_\ell \rangle \\
&\leq \langle \varphi , H_{m_1=0}  \varphi \rangle + \delta + m_1  \langle \tilde {\varphi}_\ell , N_1 \tilde {\varphi}_\ell \rangle \\
&\leq E + 2\delta + m_1 \langle \tilde {\varphi}_\ell , N_1 \tilde {\varphi}_\ell \rangle.
\end{align*}
Letting $m_1 \to 0$, we obtain that
\begin{equation*}
E^* \leq E + 2 \delta.
\end{equation*}
Hence $E^* \le E$ since $\delta>0$ is arbitrary.

To prove \eqref{eq:Em1=lim}, it suffices to observe that 
\begin{equation*}
E_{m_1=0} \le \langle \Phi_{m_1} , H_{m_1=0} \Phi_{m_1} \rangle \le \langle \Phi_{m_1} , H_{m_1} \Phi_{m_1} \rangle = E_{m_1}.
\end{equation*}
Letting $m_1\to0$ concludes the proof.
\end{proof}

The next step is to prove that $\Phi_{m_1}$ converges strongly, as $m_1 \to 0$, along some subsequence, to a non-vanishing vector of the Hilbert space which will turn out to be a ground state of $H_{m_1=0}$. An important ingredient is to control the expectation of the number operator $N_1$ in the approximate ground state $\Phi_{m_1}$, uniformly in $m_1>0$. This is the purpose of the following proposition.
\begin{proposition}
\label{Nestimate}
Let $i_0 \in \ldbrack 1 , n \rdbrack$ and $\varepsilon > 0$. Suppose that, for all $p \in \ldbrack 0 , n \rdbrack$ and all set of integers $\{ i_1,\dots,i_n \} \in \mathfrak{I}_p$,
\begin{align*}
G^{(p)}_{i_1,\dots,i_n} \in \mathcal{D} \Big ( \Big ( \prod_{ \small{ \substack{ i \in \ldbrack 2 , n \rdbrack_{i_0} \\ j \in \ldbrack 1 , 3 \rdbrack }Ê} } ( \mathrm{h}_i^{j} )^{ \frac12 - \frac{1}{n-1} + \varepsilon } \Big ) \Big ( \prod_{ \small{ \substack { i \in \ldbrack 1 , 1 \rdbrack_{i_0} \\ j \in \ldbrack 1 , 3 \rdbrack } }Ê}  ( \mathrm{h}_i^{j} )^{ \frac12 - \frac56 \frac{1}{n-1} +\varepsilon }   \Big ) \Big ) .
\end{align*}
For a.e. $\xi_1 =( k_1 , \lambda_1 ) \in \mathbb{R}^3 \times \{ -1/2 , 1/2 \}$, $k_1 \neq 0$, we have that
\begin{align}
&\big \| b_{1}(\xi_1) \Phi_{m_1} \big \| \le \mathrm{C}_0 |k_1|^{-1} \sum_{p=0}^n \sum_{i_2,\dots,i_n}  \Big \| \Big ( \prod_{ \small{ \substack{ i \in \ldbrack 2 , n \rdbrack_{i_0} \\ j \in \ldbrack 1 , 3 \rdbrack }Ê} } ( \mathrm{h}_i^{j} )^{{ \frac12 - \frac{1}{n-1} + \varepsilon }} \Big ) G^{(p)}_{1,i_2,\dots,i_n}( \xi_1 , \cdot , \dots, \cdot )  \Big \|_2 , \label{eq:number_estim}
\end{align}
where $\mathrm{C}_0$ is a positive constant independent of $m_1$ and the sum runs over all integers $i_2,\dots,i_n$ such that $\{i_2,\dots,i_n\}=\{2,\dots,n\}$, $i_2<\cdots<i_p$ and $i_{p+1}<\cdots<i_{n}$.
\end{proposition}
\begin{proof}
We have to distinguish the case where the number of different fields, $n$, is even from that where it is odd.

Suppose first that $n$ is even. Since $\Phi_{m_1}$ is a ground state of $H_{m_1}$, we have that
\begin{equation*}
b_{1}(\xi_1)(H_{m_1}-E_{m_1})\Phi_{m_1}=0 ,
\end{equation*}
and the pull-through formula then yields
\begin{align}
(H_{m_1}-E_{m_1}+ (k_1^2+m_1^2)^{\frac12}) b_{1}(\xi_1) \Phi_{m_1}  +  [ b_{1}(\xi_1),H_{I}] \Phi_{m_1}  =  0. \label{eq:csq_pull}
\end{align}
Recalling the expression \eqref{eq:defHI} of $H_I$, a direct computation shows that
\begin{align}
&[b_{1}(\xi_1),H_{I}]  \notag \\
&= \sum_{p=0}^n \sum_{i_2,\dots,i_n}  \int G^{(p)}_{1,i_2,\dots,i_n}( \xi_1 , \dots , \xi_n ) b^*_{i_2}( \xi_{i_2} ) \dots b^*_{i_p}( \xi_{i_p} ) b_{i_{p+1}}( \xi_{i_{p+1}} ) \dots b_{i_n}( \xi_{i_n} ) d \xi_2 \dots d \xi_n ,  \label{eq:commute_pull}
\end{align}
where the second sum runs over all integers $i_2,\dots,i_n$ such that $\{i_2,\dots,i_n\}=\{2,\dots,n\}$, $i_2<\cdots<i_p$ and $i_{p+1}<\cdots<i_{n}$. Applying Proposition \ref{prop:relative_bound}, we obtain that
\begin{align*}
&\big \|Ê[b_{1}(\xi_1),H_{I}] \Phi_{m_1} \big \| \\
& \le \mathrm{C}_{s,\theta,n}\sum_{p=0}^n \sum_{i_2,\dots,i_n}  \Big \| \Big ( \prod_{ \small{ \substack{ i \in \ldbrack 2 , n \rdbrack_{i_0} \\ j \in \ldbrack 1 , 3 \rdbrack }Ê} } ( \mathrm{h}_i^{j} )^{ \theta s} \Big ) G^{(p)}_{1,i_2,\dots,i_n}( \xi_1 , \cdot , \dots, \cdot )  \Big \|_2 \Big \| \Big ( \sum_{ i \in \ldbrack 2, n \rdbrack_{i_0} }H_{f,i} + 1 \Big )^{\frac{n-2}{2} (1 - \theta )} \Phi_{m_1} \Big \| ,
\end{align*}
for any $s>1/2$ and $0\le\theta\le1$, where $\mathrm{C}_{s,\theta,n}$ is a positive constant independent of $m_1>0$. Fixing $s$ and $\theta$ as in the proof of Theorem \ref{th:main}, it is then not difficult to deduce from the previous estimate that
\begin{align*}
&\big \|Ê[b_{1}(\xi_1),H_{I}] \Phi_{m_1} \big \| \\
& \le \mathrm{C} \sum_{p=0}^n \sum_{i_2,\dots,i_n}  \Big \| \Big ( \prod_{ \small{ \substack{ i \in \ldbrack 2 , n \rdbrack_{i_0} \\ j \in \ldbrack 1 , 3 \rdbrack }Ê} } ( \mathrm{h}_i^{j} )^{{ \frac12 - \frac{1}{n-1} + \varepsilon }} \Big ) G^{(p)}_{1,i_2,\dots,i_n}( \xi_1 , \cdot , \dots, \cdot )  \Big \|_2 \Big \| \Big ( \sum_{ i \in \ldbrack 2, n \rdbrack_{i_0} }H_{f,i} + 1 \Big ) \Phi_{m_1} \Big \| ,
\end{align*}
where $\mathrm{C}>0$ does not depend on $m_1$. Together with \eqref{eq:csq_pull}, since 
\begin{equation*}
\big \| (H_{m_1}-E_{m_1}+ (k_1^2+m_1^2)^{\frac12})^{-1}\big \| \le |k_1|^{-1},
\end{equation*}
we obtain that
\begin{align*}
&\big \| b_{1}(\xi_1) \Phi_{m_1} \big \| \\
& \le \mathrm{C} |k_1|^{-1} \sum_{p=0}^n \sum_{i_2,\dots,i_n}  \Big \| \Big ( \prod_{ \small{ \substack{ i \in \ldbrack 2 , n \rdbrack_{i_0} \\ j \in \ldbrack 1 , 3 \rdbrack }Ê} } ( \mathrm{h}_i^{j} )^{{ \frac12 - \frac{1}{n-1} + \varepsilon }} \Big ) G^{(p)}_{1,i_2,\dots,i_n}( \xi_1 , \cdot , \dots, \cdot )  \Big \|_2 \Big \| \Big ( \sum_{ i \in \ldbrack 2, n \rdbrack_{i_0} }H_{f,i} + 1 \Big ) \Phi_{m_1} \Big \| .
\end{align*}
Moreover, applying \eqref{eq:estim_lm1_2_4} in Remark \ref{rk:refined_rel_bound}, with $I=\{1\}$, shows that
\begin{equation*}
\big \|ÊH_I \Phi_{m_1} \big \| \le \mathrm{C}' \big \| \big ( H_{f,m_1} + 1 \big ) \Phi_{m_1} \big \| ,
\end{equation*}
for some positive constant $\mathrm{C}'$ independent of $m_1$. Combined with the previous equation and the fact that $\| H_{m_1} \Phi_{m_1} \| = |E_{m_1}|$ is uniformly bounded in $m_1\ge0$ in a compact set (since $m_1\mapsto E_{m_1}$ is non-decreasing), this yields that
\begin{align*}
&\big \| b_{1}(\xi_1) \Phi_{m_1} \big \|  \le \mathrm{C}'' |k_1|^{-1} \sum_{p=0}^n \sum_{i_2,\dots,i_n}  \Big \| \Big ( \prod_{ \small{ \substack{ i \in \ldbrack 2 , n \rdbrack_{i_0} \\ j \in \ldbrack 1 , 3 \rdbrack }Ê} } ( \mathrm{h}_i^{j} )^{{ \frac12 - \frac{1}{n-1} + \varepsilon }} \Big ) G^{(p)}_{1,i_2,\dots,i_n}( \xi_1 , \cdot , \dots, \cdot )  \Big \|_2 ,
\end{align*}
with $\mathrm{C}''>0$ independent of $m_1$. This proves \eqref{eq:number_estim} in the case where $n$ is even.

If $n$ is odd, the proof has to be modified as follows. Using anti-commutation relations, we now find that
\begin{align}
& [b_1(\xi_1),H_{I}] = -2 H_{I} b_1(\xi_1) + H'_I( \xi_1 ) , \label{eq:nodd}
\end{align}
where
\begin{align*}
H'_I(\xi_1) := \sum_{p=0}^n \sum_{i_2,\dots,i_n}  \int G^{(p)}_{1,i_2,\dots,i_n}( \xi_1 , \dots , \xi_n ) b^*_{i_2}( \xi_{i_2} ) \dots b^*_{i_p}( \xi_{i_p} ) b_{i_{p+1}}( \xi_{i_{p+1}} ) \dots b_{i_n}( \xi_{i_n} ) d \xi_2 \dots d \xi_n.
\end{align*}
The identity \eqref{eq:csq_pull} is thus replaced by
\begin{align*}
 \big ( H_{m_1} - 2 H_I - E_{m_1} + (k_1^2+m_1^2)^{\frac12} \big ) b_1(\xi_1) \Phi_{m_1} + H'_I(\xi_1) \Phi_{m_1} = 0.
\end{align*}
Now, using that $n$ is odd, a direct computation gives
\begin{align*}
(-1)^N H_{m_1} (-1)^N = H_{m_1} - 2 H_I ,
\end{align*}
where $N = \sum_{i=1}^n \int b^*_i(\xi_i)b_i(\xi_i) d\xi_i$ is the total number operator. Therefore, $H_{m_1}$ and $H_{m_1} - 2 H_I$ are unitarily equivalent and hence we have that $\inf \sigma( H_{m_1} - 2 H_I ) = E_{m_1}$. This shows that $ ( H_{m_1} - 2 H_I - E_{m_1} + (k_1^2+m_1^2)^{1/2} )$ is invertible with inverse bounded by $|k_1|^{-1}$. The rest of the proof is identical to that of the previous case.
\end{proof}
The following further technical estimate will be used in the proof of the existence of a ground state for $H_{m_1=0}$.
\begin{proposition}
\label{Nablaestimate}
Under the conditions of Proposition \ref{Nestimate}, for a.e. $\xi_1 =( k_1 , \lambda_1 ) \in \mathbb{R}^3 \times \{ -1/2 , 1/2 \}$, $k_1 \neq 0$, we have that
\begin{align}
& \big \| \nabla_{k_1} ( b_{1}(\xi_1) \Phi_{m_1} ) \big \| \le \mathrm{C}_0 |k_1|^{-2} \sum_{p=0}^n \sum_{i_2,\dots,i_n}  \Big \| \Big ( \prod_{ \small{ \substack{ i \in \ldbrack 2 , n \rdbrack_{i_0} \\ j \in \ldbrack 1 , 3 \rdbrack }Ê} } ( \mathrm{h}_i^{j} )^{{ \frac12 - \frac{1}{n-1} + \varepsilon }} \Big ) G^{(p)}_{1,i_2,\dots,i_n}( \xi_1 , \cdot , \dots, \cdot )  \Big \|_2 \notag \\
& \qquad + \mathrm{C}_0 |k_1|^{-1} \sum_{p=0}^n \sum_{i_2,\dots,i_n}  \Big \| \Big ( \prod_{ \small{ \substack{ i \in \ldbrack 2 , n \rdbrack_{i_0} \\ j \in \ldbrack 1 , 3 \rdbrack }Ê} } ( \mathrm{h}_i^{j} )^{{ \frac12 - \frac{1}{n-1} + \varepsilon }} \Big ) ( \nabla_{k_1} G^{(p)}_{1,i_2,\dots,i_n} ) ( \xi_1 , \cdot , \dots, \cdot )  \Big \|_2 , \label{eq:number_estim_grad}
\end{align}
where $\mathrm{C}_0$ is a positive constant independent of $m_1$.
\end{proposition}
\begin{proof}
Consider for instance the case where $n$ is even. We rewrite \eqref{eq:csq_pull} as
\begin{align}
 b_{1}(\xi_1) \Phi_{m_1}  = - \big (H_{m_1}-E_{m_1}+ (k_1^2+m_1^2)^{\frac12} \big )^{-1} [ b_{1}(\xi_1),H_{I}] \Phi_{m_1}  .
\end{align}
Differentiating w.r.t. $k_1$ and using \eqref{eq:commute_pull}, we obtain that
\begin{align*}
& \nabla_{k_1} ( b_{1}(\xi_1) \Phi_{m_1} ) = - \nabla_{k_1} \Big ( \big (H_{m_1}-E_{m_1}+ (k_1^2+m_1^2)^{\frac12} \big )^{-1} \Big ) \\
& \quad  \sum_{p=0}^n \sum_{i_2,\dots,i_n}  \int G^{(p)}_{1,i_2,\dots,i_n}( \xi_1 , \dots , \xi_n ) b^*_{i_2}( \xi_{i_2} ) \dots b^*_{i_p}( \xi_{i_p} ) b_{i_{p+1}}( \xi_{i_{p+1}} ) \dots b_{i_n}( \xi_{i_n} ) \Phi_{m_1} d \xi_2 \dots d \xi_n  \\
& - \big (H_{m_1}-E_{m_1}+ (k_1^2+m_1^2)^{\frac12} \big )^{-1}\\
&\quad  \sum_{p=0}^n \sum_{i_2,\dots,i_n}  \int ( \nabla_{k_1} G^{(p)}_{1,i_2,\dots,i_n} ) ( \xi_1 , \dots , \xi_n ) b^*_{i_2}( \xi_{i_2} ) \dots b^*_{i_p}( \xi_{i_p} ) b_{i_{p+1}}( \xi_{i_{p+1}} ) \dots b_{i_n}( \xi_{i_n} ) \Phi_{m_1} d \xi_2 \dots d \xi_n .
\end{align*}
Proceeding as in the proof of Proposition \ref{Nestimate}, it is not difficult to deduce from the previous equality that \eqref{eq:number_estim_grad} holds. 

The proof of \eqref{eq:number_estim_grad} in the case where $n$ is odd is analogous, using \eqref{eq:nodd} instead of \eqref{eq:csq_pull}.
\end{proof}
\begin{remark}\label{rk:bounds_l}
Proceeding in the same way as in Propositions \ref{Nestimate} and \ref{Nablaestimate}, one can estimate the norms of $b(\xi_l) \Phi_{m_1}$ and $\nabla_{k_l} ( b(\xi_l) \Phi_{m_1})$, $l \in \ldbrack 2 , n \rdbrack$, as
\begin{align*}
&\big \| b_{l}(\xi_l) \Phi_{m_1} \big \| \le \mathrm{C}_0 \omega_l(k_l)^{-1} \notag \\
&\qquad \times \sum_{p=0}^n \sum_{i_1,\dots,i_n}  \Big \| \Big ( \prod_{ \small{ \substack{ i \in \ldbrack 2 , n \rdbrack_{i_0} \\ j \in \ldbrack 1 , 3 \rdbrack }Ê} } ( \mathrm{h}_i^{j} )^{ \frac12 - \frac{1}{n-1} + \varepsilon } \Big ) \Big ( \prod_{ \small{ \substack { i \in \ldbrack 1 , 1 \rdbrack_{i_0} \\ j \in \ldbrack 1 , 3 \rdbrack } }Ê}  ( \mathrm{h}_1^{j} )^{ \frac12 - \frac56 \frac{1}{n-1} +\varepsilon }   \Big ) G^{(p)}_{i_1,i_2,\dots,i_n}( \xi_1 , \cdot , \dots, \cdot )  \Big \|_2 , 
\end{align*}
and
\begin{align*}
&\big \| \nabla_{k_l} ( b_{l}(\xi_l) ) \Phi_{m_1} \big \| \le \mathrm{C}_0 \omega_l(k_l)^{-2} \notag \\
& \times \Big \{ \sum_{p=0}^n \sum_{i_1,\dots,i_n}  \Big \| \Big ( \prod_{ \small{ \substack{ i \in \ldbrack 2 , n \rdbrack_{i_0} \\ j \in \ldbrack 1 , 3 \rdbrack }Ê} } ( \mathrm{h}_i^{j} )^{ \frac12 - \frac{1}{n-1} + \varepsilon } \Big ) \Big ( \prod_{ \small{ \substack { i \in \ldbrack 1 , 1 \rdbrack_{i_0} \\ j \in \ldbrack 1 , 3 \rdbrack } }Ê}  ( \mathrm{h}_1^{j} )^{ \frac12 - \frac56 \frac{1}{n-1} +\varepsilon }   \Big ) G^{(p)}_{i_1,i_2,\dots,i_n}( \xi_1 , \cdot , \dots, \cdot )  \Big \|_2  \\
& + \omega_l(k_l)  \sum_{p=0}^n \sum_{i_1,\dots,i_n}  \Big \| \Big ( \prod_{ \small{ \substack{ i \in \ldbrack 2 , n \rdbrack_{i_0} \\ j \in \ldbrack 1 , 3 \rdbrack }Ê} } ( \mathrm{h}_i^{j} )^{ \frac12 - \frac{1}{n-1} + \varepsilon } \Big ) \Big ( \prod_{ \small{ \substack { i \in \ldbrack 1 , 1 \rdbrack_{i_0} \\ j \in \ldbrack 1 , 3 \rdbrack } }Ê}  ( \mathrm{h}_1^{j} )^{ \frac12 - \frac56 \frac{1}{n-1} +\varepsilon }   \Big ) ( \nabla_{k_l} G^{(p)}_{i_1,i_2,\dots,i_n} ) ( \xi_1 , \cdot , \dots, \cdot )  \Big \|_2 \Big \} , 
\end{align*}
where $\omega_l(k_l) = ( k_l^2 + m_l^2 )^{1/2}$ and $\mathrm{C}_0$ is a positive constant independent of $m_1$. Theses estimates are not optimal, but they are sufficient for our purpose.
\end{remark}
We can finally prove the main theorem of this section.
\begin{theorem}
Under the conditions of Proposition \ref{Nestimate}, the operator $H_{m_1=0}$ has a ground state, i.e., there exists $\Phi_{m_1=0} \in \mathcal{D}(H_{m_1=0})$, $\Phi_{m_1=0} \neq 0$, such that
\begin{equation*}
H_{m_1=0} \Phi_{m_1=0} = E_{m_1=0} \Phi_{m_1=0}.
\end{equation*}
\end{theorem}

\begin{proof}
Let $( m_1^{(j)} )_{j\in\mathbb{N}}$ be a decreasing sequence of positive real numbers such that $m_1^{(j)} \to 0$ as $j \to \infty$. To shorten notations, we denote by $\Phi_j \in \mathcal{D}( H_{m_1^{(j)}} )$ a normalized ground state of $H_{m_1^{(j)}}$, which exists according to Proposition \ref{prop:massive}. By Proposition \ref{minimsequence}, we have that
\begin{equation*}
 \lim_{ j \to \infty } \big \| ( H_{m_1=0} - E_{m_1=0} )^{\frac12} \Phi_{j} \big \| = 0. 
\end{equation*}
To prove the theorem, we claim that it suffices to show that $(\Phi_j)$ converges strongly, as $j \to \infty$. Indeed, if we prove that there exists $\Phi_{m_1=0}$ such that $\| \Phi_j - \Phi_{m_1=0} \|$, as $j \to \infty$, we can deduce from the previous equality that $\Phi_{m_1=0} \in \mathcal{D} ( ( H_{m_1=0} - E_{m_1=0} )^{\frac12} )$, $\| \Phi_{m_1=0} \| = 1$ and $( H_{m_1=0} - E_{m_1=0} )^{\frac12}  \Phi_{m_1=0} = 0$. The statement of the theorem then follows.

Now, we prove that $(\Phi_j)$ converges strongly as $j\to \infty$. We decompose
\begin{equation*}
\Phi_j = \sum_{l_1 , \dots , l_n \in \mathbb{N}} \Phi_{j}^{(l_1,\dots,l_n)} \in \otimes_a^{l_1} L^2 ( \mathbb{R}^3 \times \{ -\frac12 , \frac12 \} ) \otimes \cdots \otimes_a^{l_n} L^2 ( \mathbb{R}^3 \times \{ -\frac12 , \frac12 \} ).
\end{equation*}
Recall that $N$ stands for the total number operator in $\mathcal{H}$. By Proposition \ref{Nestimate} and Remark \ref{rk:bounds_l}, using Hypothesis \eqref{eq:cond_GS_1}, one can see that $\langle \Phi_j , N \Phi_j \rangle$ is uniformly bounded in $j \in \mathbb{N}$. From this property, one can deduce that the strong convergence of $\Phi_{j}^{(l_1,\dots,l_n)}$ for any $l_1,\dots,l_n \in \mathbb{N}$ implies the strong convergence of $\Phi_j$.

Moreover, we claim that it suffices to prove the $L^2$ convergence of $\Phi_{j}^{(l_1,\dots,l_n)}$ on any compact subset of $( \mathbb{R}^3 \times \{ -1/2 , 1/2 \} )^{l_1+\dots+l_n}$. Indeed, we observe that
\begin{align*}
& \Phi_{j}^{(l_1,\dots,l_n)}( \xi_1^{(1)} , \dots , \xi_1^{(l_1)} , \dots , \xi_n^{(1)} , \dots , \xi_n^{(l_n)} ) \\
& = \frac{1}{\sqrt{l_1}} b_1(\xi_1^{(1)}) \Phi_j^{(l_1-1,l_2,\dots,l_n)} ( \xi_1^{(2)} , \dots , \xi_1^{(l_1)} , \dots , \xi_n^{(1)} , \dots , \xi_n^{(l_n)} ) .
\end{align*}
Recall that $B_\Lambda = \{ ( k , \lambda ) \in \mathbb{R}^3 \times \{ -1/2 , 1/2 \} , |k| \le \Lambda \}$. From Proposition \ref{Nestimate}, it follows that
\begin{align*}
&\big \| \mathds{1}_{B_\Lambda^{\mathrm{c}} }( \xi_1^{(1)} ) b_{1}(\xi_1) \Phi_{j} \big \| \le \frac{ \mathrm{C}_0 }{ \Lambda } \sum_{p=0}^n \sum_{i_2,\dots,i_n}  \Big \| \Big ( \prod_{ \small{ \substack{ i \in \ldbrack 2 , n \rdbrack_{i_0} \\ j \in \ldbrack 1 , 3 \rdbrack }Ê} } ( \mathrm{h}_i^{j} )^{{ \frac12 - \frac{1}{n-1} + \varepsilon }} \Big ) G^{(p)}_{1,i_2,\dots,i_n}( \xi_1 , \cdot , \dots, \cdot )  \Big \|_2 ,
\end{align*}
uniformly in $j \in \mathbb{N}$. This implies that
\begin{align}
\big \| \mathds{1}_{B_\Lambda^{\mathrm{c}} }( \cdot ) \Phi_{j}^{(l_1,\dots,l_n)} \big \|_2^2 \lesssim \frac{1}{\Lambda^2} \sum_{p=0}^n \sum_{i_2,\dots,i_n}  \Big \| \Big ( \prod_{ \small{ \substack{ i \in \ldbrack 2 , n \rdbrack_{i_0} \\ j \in \ldbrack 1 , 3 \rdbrack }Ê} } ( \mathrm{h}_i^{j} )^{{ \frac12 - \frac{1}{n-1} + \varepsilon }} \Big ) G^{(p)}_{1,i_2,\dots,i_n}  \Big \|_2^2 . \label{eq:compact_enough}
\end{align}
Here $\mathds{1}_{B_\Lambda^{\mathrm{c}} }( \cdot ) \Phi_{j}^{(l_1,\dots,l_n)}$ should be understood as the map
\begin{equation*}
\big ( \xi_1^{(1)} , \dots , \xi_1^{(l_1)} , \dots , \xi_n^{(1)} , \dots , \xi_n^{(l_n)} \big ) \mapsto \mathds{1}_{B_\Lambda^{\mathrm{c}} }( \xi_1^{(1)} ) \Phi_{j}^{(l_1,\dots,l_n)} \big ( \xi_1^{(1)} , \dots , \xi_1^{(l_1)} , \dots , \xi_n^{(1)} , \dots , \xi_n^{(l_n)} \big ).
\end{equation*}
Clearly, the right hand side of \eqref{eq:compact_enough} can be made arbitrarily small, uniformly in $j \in \mathbb{N}$, for $\Lambda$ large enough. An analogous estimate holds if $\xi_1^{(1)}$ is replaced by any of the other variable $\xi_i^{(\ell)}$. This shows that it suffices to establish the $L^2$ convergence of $\Phi_{j}^{(l_1,\dots,l_n)}$ on any compact subset of $( \mathbb{R}^3 \times \{ -1/2 , 1/2 \} )^{l_1+\dots+l_n}$.

Now we prove that $( \Phi_{j}^{(l_1,\dots,l_n)} )_{j \in \mathbb{N}} $ converges strongly on 
\begin{equation*}
B_\Lambda^{l_1+\dots+l_n} \subset ( \mathbb{R}^3 \times \{ -1/2 , 1/2 \} )^{l_1+\dots+l_n}.
\end{equation*}
We note that, since $\| \Phi_j \|Ê= 1$ for all $j \in \mathbb{N}$, $(\Phi_{j}^{(l_1,\dots,l_n)})_{j\in\mathbb{N}}$ is a bounded sequence in $L^2$, hence, by the Banach-Alaoglu theorem, it converges weakly in $L^2$, along some subsequence. Consider a weakly convergent subsequence which, to simplify, is denoted by the same symbol $(\Phi_{j}^{(l_1,\dots,l_n)})_{j\in\mathbb{N}}$. We define the space $W^{1,r}( B_\Lambda^{l_1+\dots+l_n} )$ as the set of all measurable maps $f$ from $B_\Lambda^{l_1+\dots+l_n}$ to $\mathbb{C}$ such that, for all values of the spin variables $(\lambda_1^{(1)} , \dots , \lambda_n^{(l_n)} )$, the map
\begin{equation*}
f \big ( \cdot , \lambda_1^{(1)} , \cdot , \lambda_1^{(2)} , \dots , \cdot , \lambda_n^{(l_n)} \big )
\end{equation*}
belongs to the Sobolev space $W^{1,r}( \{ |k| \le \Lambda \}^{l_1+\dots+l_n} )$.  We claim that, for all $j \in \mathbb{N}$, $\Phi_{j}^{(l_1,\dots,l_n)}$ belongs to $W^{1,r}( B_\Lambda^{l_1+\dots+l_n} )$ provided that $1\le r<2$. Indeed, since $r<2$ and $\Phi_{j}^{(l_1,\dots,l_n)} \in L^2 ( B_\Lambda^{l_1+\dots+l_n} )$, we see that $\Phi_{j}^{(l_1,\dots,l_n)} \in L^r ( B_\Lambda^{l_1+\dots+l_n} )$. Moreover, similarly as in \cite{GrLiLo01_01}, we can compute
\begin{align*}
&\int_{B_\Lambda^{l_1+\dots+l_n}} \Big | ( \nabla_{k_1^{(1)}} \Phi_{j}^{(l_1,\dots,l_n)} ) \big ( \xi_1^{(1)} , \dots , \xi_n^{(l_n)} \big ) \Big |^r d\xi_1^{(1)} \dots d\xi_n^{(l_n)} \\
&= \frac{1}{\sqrt{l_1}^r} \int_{B_\Lambda^{l_1+\dots+l_n}} \Big | ( \nabla_{k_1^{(1)}} b_1(\xi_1^{(1)} ) \Phi_j^{(l_1-1,l_2,\dots,l_n)} ) \big ( \xi_1^{(2)} , \dots , \xi_n^{(l_n)} \big ) \Big |^r d\xi_1^{(1)} \dots d\xi_n^{(l_n)} \\
&\lesssim \int_{B_\Lambda} \Big ( \int_{B_\Lambda^{(l_1-1)+\dots+l_n}} \Big | ( \nabla_{k_1^{(1)}} b_1(\xi_1^{(1)} ) \Phi_j^{(l_1-1,l_2,\dots,l_n)} ) \big ( \xi_1^{(2)} , \dots , \xi_n^{(l_n)} \big ) \Big |^2 d\xi_1^{(2)} \dots d\xi_n^{(l_n)} \Big )^{\frac{r}{2}} d \xi_1^{(1)} \\
&\le \int_{B_\Lambda} \big \| \nabla_{k_1^{(1)}} ( b_1(\xi_1^{(1)} ) \Phi_j ) \big \|^r d \xi_1^{(1)} ,
\end{align*}
the first inequality being a consequence of H{\"o}lder's inequality. Applying Proposition \ref{Nablaestimate}, we obtain that
\begin{align*}
&\int_{B_\Lambda^{l_1+\dots+l_n}} \Big | ( \nabla_{k_1^{(1)}} \Phi_{j}^{(l_1,\dots,l_n)} ) \big ( \xi_1^{(1)} , \dots , \xi_n^{(l_n)} \big ) \Big |^r d\xi_1^{(1)} \dots d\xi_n^{(l_n)} \\
&\lesssim \sum_{p=0}^n \sum_{i_2,\dots,i_n} \int_{B_\Lambda}  \big|k_1^{(1)}\big|^{-2r}  \Big \| \Big ( \prod_{ \small{ \substack{ i \in \ldbrack 2 , n \rdbrack_{i_0} \\ j \in \ldbrack 1 , 3 \rdbrack }Ê} } ( \mathrm{h}_i^{j} )^{{ \frac12 - \frac{1}{n-1} + \varepsilon }} \Big ) G^{(p)}_{1,i_2,\dots,i_n}( \xi_1^{(1)} , \cdot , \dots, \cdot )  \Big \|_2^r d \xi_1^{(1)} \\
&\qquad + \sum_{p=0}^n \sum_{i_2,\dots,i_n} \int_{B_\Lambda}  \big | k_1^{(1)} \big |^{-r}  \Big \| \Big ( \prod_{ \small{ \substack{ i \in \ldbrack 2 , n \rdbrack_{i_0} \\ j \in \ldbrack 1 , 3 \rdbrack }Ê} } ( \mathrm{h}_i^{j} )^{{ \frac12 - \frac{1}{n-1} + \varepsilon }} \Big ) ( \nabla_{k_1^{(1)}} G^{(p)}_{1,i_2,\dots,i_n}( \xi_1^{(1)} , \cdot , \dots, \cdot ) )  \Big \|_2^r d \xi_1^{(1)} ,
\end{align*}
which is finite by assumption. In the same way, one can verify that the other derivatives $\nabla_{k_i^{(\ell)}} \Phi_{j}^{(l_1,\dots,l_n)}$ belong to $L^r$. Hence $\Phi_{j}^{(l_1,\dots,l_n)} \in W^{1,r} ( B_\Lambda^{l_1+\dots+l_n} )$.

Finally, since $(\Phi_{j}^{(l_1,\dots,l_n)})_{j\in\mathbb{N}}$ converges weakly in $L^2( B_\Lambda^{l_1+\dots+l_n} )$, it also converges weakly in $W^{1,r}( B_\Lambda^{l_1+\dots+l_n} )$, because $r < 2$. As in \cite{GrLiLo01_01}, applying the Rellich-Kondrachov theorem, one then obtain that $(\Phi_{j}^{(l_1,\dots,l_n)})_{j\in\mathbb{N}}$ converges strongly in $L^2( B_\Lambda^{l_1+\dots+l_n} )$. This concludes the proof of the theorem.
\end{proof}

\subsection{Proof of Theorem \ref{th:GS}}\label{subsec:massless}

In order to conclude the proof of Theorem \ref{th:GS}, it suffices to proceed by induction in the following way. The induction hypothesis $(\mathrm{H}_l)$ is that, if $l$ particles are massless and $n-l$ are massive, then $H$ has a ground state. It has been shown in Section \ref{subsec:massive} that $(\mathrm{H}_0)$ holds. Assuming that $(\mathrm{H}_l)$ holds true, we proceed to prove that $(\mathrm{H}_{l+1})$ holds exactly in the same way as in Section \ref{subsec:one_massless}. More precisely, assuming that $m_1=m_2=\dots=m_{l+1}=0$ and $\min_{i\in\ldbrack l+2 , n \rdbrack} m_i > 0$, the total Hamiltonian
\begin{equation*}
H \equiv H_{m_1=0,\dots,m_l=0,m_{l+1}=0} 
\end{equation*}
is approximated by the family of operators
\begin{equation*}
H_{m_1=0,\dots,m_l=0,m_{l+1}} , \quad m_{l+1}>0 ,
\end{equation*}
where the free energy for the $(l+1)^{\text{th}}$ massless field, $\mathrm{d}\Gamma( |k_{l+1}| )$, is replaced by $\mathrm{d}\Gamma( ( k_{l+1}^2 + m_{l+1}^2 )^{1/2}$. By the induction hypothesis, $H_{m_1=0,\dots,m_l=0,m_{l+1}}$ has a normalized ground state $\Psi_{m_1=0,\dots,m_l=0,m_{l+1}}$. One shows that $\Psi_{m_1=0,\dots,m_l=0,m_{l+1}}$ converges strongly, as $m_{l+1} \to 0$, to a ground state of $H_{m_1=0,\dots,m_l=0,m_{l+1}=0}$ by adapting the proof given in Section \ref{subsec:one_massless} in a straightforward way. Details are left to the reader.


\noindent \textbf{Acknowledgments} J.F. and J.-C.G. are grateful to J.-M. Barbaroux for many discussions and fruitful collaborations.

\bibliographystyle{amsalpha}

\end{document}